\documentclass[11pt]{article}
\usepackage{amssymb,amsmath,verbatim,amsthm}
\newtheorem{Theorem}{Theorem}[section]
\newtheorem{Lemma}[Theorem]{Lemma}
\newtheorem{Proposition}[Theorem]{Proposition}

\newtheorem{Corollary}[Theorem]{Corollary}
\newtheorem{Remark}[Theorem]{Remark}
\newtheorem{Example}[Theorem]{Example}
\newtheorem{Definition}[Theorem]{Definition}

\usepackage{enumerate}
\usepackage{color,xcolor}
\usepackage{indentfirst}
\usepackage{booktabs}
\usepackage{makecell}
\usepackage{array}
\usepackage{bm}
\newcommand{\ignore}[1]{}

\newcommand{\PGL}{{\operatorname{PGL}}}

\newcommand{\Aut}{{\operatorname{Aut}}}
\newcommand{\MAut}{{\operatorname{MAut}}}
\newcommand{\PAut}{{\operatorname{PAut}}}
\newcommand{\GDD}{{\operatorname{GDD}}}
\newcommand{\Stab}{{\operatorname{Stab}}}
\newcommand{\Tr}{{\operatorname{Tr}}}

\begin{document}
	%------------------------------------------------------------------------
	\renewcommand{\theequation}{\thesection.\arabic{equation}}
	
	\title{\bf Research on Linear Codes Holding $q$-Ary $t$-Designs\footnote{Supported by National Natural Science Foundation under Grant 12571346. (Corresponding Author: Junling Zhou. Email: {\tt jlzhou@bjtu.edu.cn})}}
	\author{
		{ Xinghao Wu\footnote{School of Mathematics and Statistics, Beijing Jiaotong University,  Beijing  100044, China. (Email: {\tt 23111512@bjtu.edu.cn})}}, { Junling Zhou\footnote{School of Mathematics and Statistics, Beijing Jiaotong University,  Beijing  100044, China; Hebei Provincial Key Laboratory of Mathematical Theory and Analysis for Network and Data Science.
		}}\\
	}
	\date{ }
	\maketitle
	%------------------------------------------------------------------------	
	\begin{abstract}
		A $q$-ary $t$-$(n,w,\lambda)$ design is a collection $\mathcal{A}$ of vectors of weight $w$ in $\mathbb{F}_{q}^{n}$ with the property that  every vector of weight $t$ in $\mathbb{F}_{q}^{n}$ is contained in exactly $\lambda$ members of $\mathcal{A}$.
		The supports of the vectors in a $q$-ary $t$-design form an ordinary $t$-design, possibly with repeated blocks. While linear codes supporting ordinary combinatorial designs have been extensively studied, the case where codes hold $q$-ary designs remains largely unexplored. This motivates a systematic investigation into whether codewords of a fixed weight in a linear code can form a $q$-ary $t$-design. Building on previous work, we develop two new criteria for this purpose. Applying these criteria, we show that several families of linear codes hold $q$-ary $2$-designs, including one- and two-weight codes, extremal self-dual codes, as well as certain dual codes, shortened codes, and punctured codes derived from them.
		Moreover, for linear codes that do not satisfy these criteria, we provide an alternative approach based on the automorphism group of the code. This method enables the construction of $q$-ary $2$-designs from doubly-extended Reed-Solomon codes. Notably, for a class of linear codes previously known to support $4$-designs, we demonstrate that their codewords of certain weights give rise to $q$-ary $2$-designs whose parameters are precisely determined.	
		\medskip
		
		\noindent {\bf Keywords}: $t$-design, $q$-ary $t$-design, linear code, doubly-extended Reed-Solomon code, MDS code, trace code.
		\medskip
	\end{abstract}
	%------------------------------------------------------------------------	
	\section{Introduction}
		
	Design theory and coding theory are intimately connected, with each field elegantly drawing upon the results and methodologies of the other. On one hand, the incidence matrix of a $t$-design is applied to generate a linear code over the finite field $\mathbb{F}_q$. By this method many good codes have been produced and documented in the literature (see, for example, \cite{Key, Ding-DS, Tonchev1, Tonchev2}). On the other hand, one may obtain a $t$-design by taking the supports of all codewords of a fixed weight in a code. Our focus lies  primarily in the latter. However, we consider all codewords of a fixed weight (rather than just their supports) and investigate whether they can give rise to a $q$-ary $t$-design. If so, the supports of these codewords form an ordinary $t$-design (possibly with repeated blocks). In the special case where $q=2$, a binary  $t$-design is equivalent to an ordinary $t$-design. Numerous linear and nonlinear codes have been explored for constructing ordinary $t$-designs \cite{Cameron, Ding-Steiner, Ding-five, Ding, Pless}. In recent years, this line of research has attracted growing interest (see, for example, \cite{Awada, Cunsheng2020Infinite, Heng, Sole, Xiang, Cao1, Yan}). Notably, in 2021, Tang and Ding achieved a breakthrough by discovering an infinite family of 4-designs supported by linear codes \cite{BCH-4-design}.
	
	In 1973, Delsarte \cite{dsd's'} introduced four fundamental parameters of a code $\mathcal{C}$, which for a linear code reduce to the minimum weight and the number of distinct nonzero weights of $\mathcal{C}$ and its dual. He revealed the combinatorial significance of these parameters and introduced the notion of a $q$-ary $t$-design (referred to as a ``$t$-design of type $\lambda = q-1$"). Delsarte \cite{Delsarte1973} also presented a method for identifying $q$-ary designs from codes, as stated in Lemma \ref{d'-s} of this paper. Moreover, he explored the association schemes of coding theory from an algebraic perspective \cite{Delsarte1973}, investigating, among other things, the regularity of codes.
	Shortly afterwards, Goethals and van Tilborg \cite{Tilborg1975} and van Tilborg in his PhD thesis \cite{UP} established that $q$-ary $t$-designs can be held by certain $t$-regular codes; see Lemma \ref{d-s'} herein. Consolidating and refining these theoretical developments, we propose the so-called Standard Criterion in Theorem \ref{Standard-Criterion}, which can be regarded as an analog of the well-known Assmus-Mattson theorem. In 1976, Assmus, Goethals, and Mattson \cite{punctured-t} further investigated the combinatorial properties of $q$-ary $t$-designs. Building on their work, we introduce a criterion for identifying $q$-ary designs from punctured and shortened codes, presented as the Puncturing-Shortening Criterion in Theorem \ref{punctured}.
	Following the contributions of Delsarte, Goethals and van Tilborg, and Assmus, Goethals and Mattson, no significant progress was made on the problem of constructing $q$-ary $t$-designs from linear codes for $q > 2$ for over half a century.
	This prolonged stagnation motivates us to undertake a systematic investigation of this problem.
	
	The Standard Criterion and the Puncturing-Shortening Criterion together form a general framework for generating $q$-ary $t$-designs from linear codes when $q > 2$. Applying this framework systematically to prominent families of linear codes, we demonstrate that MDS codes, perfect codes, one- and two-weight codes, and extremal self-dual codes yield a rich variety of $q$-ary designs.
	Recognizing that this general framework does not always guarantee the maximal achievable strength $t$, we further develop a new method. This approach effectively exploits the $t$-transitivity of a code's automorphism group to identify $q$-ary $t$-designs, reducing the verification to checking a single set of $t$ coordinates (see Theorem \ref{transitive-design}). This technique enables us to produce infinite families of $q$-ary $2$-designs from doubly-extended Reed-Solomon codes and a family of trace codes, in cases where the general criteria are not sufficiently sharp.	
	
	The remainder of this paper is organized as follows. Section $2$ introduces	the necessary concepts and preliminary results on linear codes, $t$-designs, $q$-ary $t$-designs, and their relationships.	 Section 3 develops two general criteria and  establishes a framework for producing $q$-ary $t$-designs from linear codes. Section 4 presents the systematic applications of the general framework to classical code families. Section 5 introduces the novel automorphism-group-based method and its applications. Section 6 summarizes and concludes the paper.

	%------------------------------------------------------------------------
	\section{Preliminaries}
	
	\subsection{Linear codes and fundamental parameters}
	Let $q$ be a prime power and $\mathbb{F}_{q}$ be the finite field with $q$ elements. Let $\mathbb{F}_{q}^{n}$ be the $n$-dimensional vector space over $\mathbb{F}_{q}$.
	For a positive integer $n$, denote $[n] := \{1,2,\dots,n\}$.
	A vector $\bm{x}$ in $\mathbb{F}_{q}^{n}$ is denoted by $\bm{x} = (x_{1},x_{2},\dots,x_{n})$.
	For vectors $\bm{x},\bm{y}\in\mathbb{F}_{q}^{n}$, the \emph{(Hamming) distance} $d(\bm{x},\bm{y})$ is the number of coordinates in which they differ, i.e., $ d(\bm{x},\bm{y})= |\{i\in [n]:x_{i}\ne y_{i}\}|$.
	The \emph{(Hamming) weight} $\mathrm{wt}(\bm{x})$ of a vector $\bm{x}$ is its distance to the zero vector. Thus $\mathrm{wt}(\bm{x})$ is the number of nonzero entries in $\bm{x}$.
	The \emph{support} of the vector $\bm{x}$ is the set of coordinate positions where its entries are nonzero, i.e., supp$(\bm{x})=\{i\in[n] : x_{i}\ne0\}$.
	A \emph{linear code} $\mathcal{C}$ of length $n$ over the finite field $\mathbb{F}_q$ is a subspace of $\mathbb{F}_{q}^{n}$.
	An $[n,k]_q$ code denotes a linear code over $\mathbb{F}_{q}$ of length $n$ and dimension $k$. If the \emph{minimum distance} $d$ is specified, then the notation $[n,k,d]_q$ is used.
	Given a code $\mathcal{C}\subseteq \mathbb{F}_{q}^{n}$ and a vector $\bm{x}\in \mathbb{F}_{q}^{n}$, the distance from $\bm{x}$ to $\cal C$ is $d(\bm{x},\mathcal{C}) = \min\{d(\bm{x},\bm{c}) : \bm{c}\in\mathcal{C} \}$.

	With respect to the standard inner product $\langle \bm{x}, \bm{y} \rangle = \sum_{i=1}^{n}x_{i}y_{i}$ on $\mathbb{F}_{q}^{n}$, the \emph{dual code} of $\mathcal{C}$ is defined as
	\[
	\mathcal{C}^{\perp} = \{\bm{x}\in \mathbb{F}_{q}^{n}: \langle \bm{x}, \bm{y} \rangle=0 \text{ for all } \bm{y}\in \mathcal{C} \}.
	\]
	$\mathcal{C}$ is called \emph{self-orthogonal} if $\mathcal{C}\subseteq\mathcal{C}^{\perp}$, and \emph{self-dual} if $\mathcal{C}=\mathcal{C}^{\perp}$.
	When working over $\mathbb{F}_{q^{2}}$, it is often useful to adopt the Hermitian inner product $\langle \bm{x}, \bm{y} \rangle_{H} = \sum_{i=1}^{n}x_{i}y_{i}^{q}$. The \emph{Hermitian dual} of $\mathcal{C}$ is defined as
	\[
	\mathcal{C}^{\perp_{H}} = \{\bm{x}\in \mathbb{F}_{q^2}^{n}: \langle \bm{x}, \bm{y} \rangle_{H}=0 \text{ for all } \bm{y}\in \mathcal{C} \}.
	\]
	Accordingly, $\cal C$ is \emph{Hermitian self-orthogonal} if $\mathcal{C}\subseteq\mathcal{C}^{\perp_{H}}$, and \emph{Hermitian self-dual} if $\mathcal{C}=\mathcal{C}^{\perp_{H}}$.

	For an $[n,k,d]_{q}$ linear code $\mathcal{C}$, let $W(\mathcal{C})$ denote the set of its distinct nonzero weights, and let $s=s(\mathcal{C})$ be the number of nonzero weights.
	The minimum distance of the dual code $\mathcal{C}^{\perp}$, denoted by $d^\perp = d(\mathcal{C}^{\perp})$, is called the \emph{dual distance} of $\cal C$. The number of nonzero weights in $\mathcal{C}^{\perp}$, written as $s^{\perp} = s(\mathcal{C}^{\perp})$, is referred to as the \emph{external distance} of $\cal C$ \cite{dsd's'}.
	The \emph{packing radius} of $\mathcal{C}$ is defined as $e = e(\mathcal{C}) = \lfloor (d-1)/2 \rfloor$, derived from the sphere-packing bound. The \emph{covering radius} $\rho = \rho(\mathcal{C})$ is the smallest integer $\rho$ satisfying $|\mathcal{C}|\sum_{i=0}^{\rho} \binom{n}{i}(q-1)^i \geq q^{n}$.
	By definition, $e \le \rho$. The code is called \emph{perfect} if equality holds, meaning that every vector in $\mathbb{F}_q^n$ is within distance $e$ from exactly one codeword.
	These parameters are interrelated by known bounds. In particular, from \cite[Theorem 3.3]{dsd's'}, we have
	\begin{align}\label{e<rho<s'}
		e\le \rho\le s^{\perp}.
	\end{align}

	A code $\mathcal{C}$ with parameters $[n,k,d]_{q}$ is called \emph{maximum distance separable} (MDS) if it meets the Singleton bound, i.e., if $d=n-k+1$.
	The code $\mathcal{C}$ is said to be \emph{divisible} by an integer $\Delta>1$ if all codewords have weights divisible by $\Delta$; in this case, $\Delta$ is called a \emph{divisor} of $\cal C$.
	
	Let $\cal C$ be a linear $[n,k,d]_{q}$ code.
	The \emph{punctured code} $\mathcal{C}^{\{m\}}$ is obtained by deleting the $m$-th coordinate from every codeword.
	The \emph{shortened code} $\mathcal{C}_{\{m\}}$
	consists of all codewords that are zero at the $m$-th coordinate, with the coordinate removed.
	For the corresponding dual operations, the following hold \cite[Theorem 1.5.7]{fundamentals}:
	$(\mathcal{C^{\perp}})_{\{m\}} = (\mathcal{C}^{\{m\}})^{\perp},
	(\mathcal{C^{\perp}})^{\{m\}} = (\mathcal{C}_{\{m\}})^{\perp}.$

	\subsection{$t$-Designs and $q$-ary $t$-designs}
	
	\begin{Definition}
		A $t\text{-}(n,w,\lambda)$ \emph{design} is an incidence structure $(\mathcal{P},\mathcal{B})$, where $\mathcal{P}$ is a set of $n$ elements (called \emph{points}), and $\mathcal{B}$ is a collection of $w$-subsets of $\mathcal{P}$ (called \emph{blocks}), such that every $t$-subset of $\mathcal{P}$ is contained in exactly $\lambda$ blocks.
		The parameter $t$ is said to be the \emph{strength} and $\lambda$ is the \emph{index} of the design.
	\end{Definition}
	For simplicity, we usually denote a $t$-design solely by its block set $\mathcal{B}$.
	The design is termed \emph{simple} if $\mathcal{B}$ contains no repeated blocks.
	Denote by $\binom{\mathcal{P}}{w}$ the set of all $w$-subsets of $\cal P$. Then $\binom{\mathcal{P}}{w}$ forms a trivial $w$-$(n,w,1)$ design, called the \emph{complete design}.
	It is well-known that a $t$-$(n,w,\lambda)$ design is also an $i$-$(n,w,\lambda_{i})$ design for $0\le i\le t$, where
	\begin{align}\label{lambda-i}
		\lambda_{i}=\lambda\binom{n-i}{t-i}\bigg/\binom{w-i}{t-i}.
	\end{align}
	
	The concept of a $t$-design can be generalized to the $q$-ary setting \cite{punctured-t,dsd's',Delsarte1973,NewCR}.
	To make a distinction, a $t$-design defined above is called an \emph{ordinary} or a \emph{classical} $t$-design.
	For any two vectors $\bm{a},\bm{b} \in \mathbb{F}^{n}_{q}$ with $\bm{a}=(a_{1},a_{2},\ldots,a_{n}), \bm{b}=(b_{1},b_{2},\ldots,b_{n})$, we say that $\bm{a}$ \emph{covers} $\bm{b}$ if for every nonzero $b_{i}$, it holds that $a_{i}=b_{i}$.
	%q-ary t-design定义 √
	\begin{Definition}\label{def-qtdesign}
		Let $\mathcal{A}$ be a nonempty set of vectors (also called blocks) in $\mathbb{F}_{q}^{n}$, all of which have a constant  weight $w$.
		Then $\mathcal{A}$ is called a \emph{$q$-ary $t$-$(n,w,\lambda)$ design}, also denoted as a $t$-$(n,w,\lambda)_{q}$ design, if for every vector $\bm{x} \in \mathbb{F}_{q}^{n}$ of weight $t$, there are exactly $\lambda$ blocks in $\mathcal{A}$ that cover $\bm{x}$.
		The parameters $t$ and $\lambda$ are referred to as strength and index respectively.
	\end{Definition}

	We now explore combinatorial properties of $q$-ary $t$-designs.
	Like ordinary $t$-designs, the existence of a $q$-ary $t$-design implies the existence of $q$-ary designs of lower strength. This is formalized in the following lemma.	
	
	%q-ary t-design的λi
	\begin{Lemma}
		\cite[Lemma 2.4.6]{UP} \label{q-design-li}
		Let $\mathcal{A}$ be a $q$-ary $t$-$(n, w, \lambda)$ design. Then, for every integer $i$ with $0 \le i \le t$, $\mathcal{A}$ also forms a $q$-ary $i$-$(n, w, \lambda_i)$ design, where
		\begin{equation}\label{q-lambda-i}
			\lambda_{i} = \lambda(q-1)^{t-i} \binom{n-i}{t-i} \bigg/\binom{w-i}{t-i}.
		\end{equation}
	\end{Lemma}
	
%	From \eqref{lambda-i}, a necessary condition for the existence of a $t$-$(n,w,\lambda)$ design is that
%	\begin{align}
%		\left. \binom{w-i}{t-i} \right | \left. \lambda\binom{n-i}{t-i}\right., \ 0\le i\le t;
%	\end{align}
%	while for a $q$-ary $t$-$(n,w,\lambda)$ design the necessary condition becomes by \eqref{q-lambda-i}
%	\begin{align}
%		\left. \binom{w-i}{t-i} \right | \left. \lambda(q-1)^{t-i}\binom{n-i}{t-i}\right., \ 0\le i\le t.
%	\end{align}
	
	A natural connection between $q$-ary designs and classical designs is established by considering the supports of the vectors.

	\begin{Proposition}\cite[Lemma 2.2]{punctured-t}\label{qt->t}
		If $\mathcal{A}\subseteq \mathbb{F}_{q}^{n}$ forms a $q$-ary $t$-$(n,w,\lambda)$ design, then the collection of supports of the vectors in $\mathcal{A}$ forms a $t$-$(n,w,(q-1)^{t}\lambda)$ design (maybe non-simple).
	\end{Proposition}
	\begin{proof}
		Let $S \subseteq \{1,2,\dots,n\}$ be an arbitrary $t$-element subset.
		Consider all vectors $\bm{x} \in \mathbb{F}_{q}^{n}$ with $\mathrm{supp}(\bm{x}) = S$.
		There are exactly $(q-1)^{t}$ such vectors.
		By the definition of a $q$-ary $t$-design, each such $\bm{x}$ is covered by exactly $\lambda$ vectors in $\mathcal{A}$.
		Thus, the number of vectors in $\mathcal{A}$ whose support contains $S$ is exactly $(q-1)^{t}\lambda$.
		Therefore, the collection of supports forms a $t$-$(n,w,(q-1)^{t}\lambda)$ design.
	\end{proof}
	
	\begin{Corollary}\label{2-ary-t-design}
		A binary $t$-$(n,w,\lambda)$ design is equivalent to a classical $t$-$(n,w,\lambda)$ design.
	\end{Corollary}

	From the perspective of design theory, $q$-ary designs can be identified with a type of group divisible designs.
	\begin{Definition}
		\label{GDD} A \emph{group divisible design} $\GDD$$_{\lambda}(t,w,gn)$ of type $g^{n}$ is a triple $(X,\mathcal{G},\mathcal{A})$ satisfying following conditions:
		\begin{itemize}
			\item [(1)] $X$ is a set of $gn$ elements called points;
			\item [(2)] $\mathcal{G}$ is a partition of $X$ into $n$ subsets (called groups), each of cardinality $g$;
			\item [(3)] $\mathcal{A}$ is a collection of subsets of $X$ called blocks where each $A \in \mathcal{A}$ has cardinality $w$;
			\item [(4)] a group and a block intersect in at most one point;
			\item [(5)] any $t$-subset of $X$ from $t$ distinct groups is contained in exactly $\lambda$ blocks.
		\end{itemize}
	\end{Definition}
	It is obvious that a $\GDD_{\lambda}(t,w,n)$ of type $1^{n}$ is the same as a $t$-$(n,w,\lambda)$ design.
	A fundamental equivalence is clear between $q$-ary designs and group divisible designs, as captured by the following proposition.
	\begin{Proposition}\label{qt->GDD}
		The existence of a $q$-ary $t$-$(n,w,\lambda)$ design is equivalent to the existence of a
		${\GDD_{\lambda}(t,w,(q-1)n)}$ of type $(q-1)^{n}$.
	\end{Proposition}
%	\begin{Proposition}\label{qt->GDD}
%		There exists a $q$-ary $t$-$(n,w,\lambda)$ design if and only if there exists a
%		${\GDD_{\lambda}(t,w,(q-1)n)}$ of type $(q-1)^{n}$.
%	\end{Proposition}
	\begin{proof}
		Let $\mathcal{A}\subseteq\mathbb{F}_{q}^{n}$ be a $q$-ary $t$-$(n,w,\lambda)$ design.
		Define the point set $X=[n]\times \mathbb{F}_{q}^{*}$, where $\mathbb{F}_{q}^{*}=\mathbb{F}_{q}\setminus\{0\}$. Let $\mathcal{G} = \{G_{i} : i\in [n]\}$, where $G_{i}=\{(i,x) : x\in\mathbb{F}_{q}^{*} \}$.
		For each vector $\bm{v}=(v_{1},\dots,v_{n})\in \mathcal{A}$, construct the block $B_{\bm{v}}=\{(i,v_{i}) : v_{i}\ne 0\}$ and take $\mathcal{B} =\{B_{\bm{v}} : \bm{v}\in\mathcal{A}\}$.
		It is immediate that the resulting structure $(X,\mathcal{G},\mathcal{B})$ is a $\GDD$$_{\lambda}(t,w,(q-1)n)$ of type $(q-1)^{n}$.
		The converse follows by reversing the above construction: each block of a $\GDD_{\lambda}(t,w,(q-1)n)$ of type $(q-1)^{n}$ over $X$ with group set $\mathcal{G}$ corresponds uniquely to a vector in $\mathbb{F}_{q}^{n}$ of weight $w$, producing a $q$-ary $t$-$(n,w,\lambda)$ design.
	\end{proof}
	
	\subsection{Designs from linear codes}
	Let $\mathcal{C}$ be an $[n,k,d]_{q}$ linear code.
	The \emph{weight enumerator} of $\mathcal{C}$ is the polynomial $A(z) = 1+A_{1}z+A_{2}z+\dots+A_{n}z^{n}$, where $A_w$ is the number of codewords of weight $w$.	
	We denote by $\mathcal{A}_{w}(\mathcal{C})$ the set of all codewords of weight $w$ in $\mathcal{C}$ and thus $|\mathcal{A}_{w}(\mathcal{C})|=A_{w} $.
	If the set $\mathcal{A}_{w}(\mathcal{C})$ of codewords of weight $w$ itself forms a $q$-ary $t$-$(n, w, \lambda)$ design, we say that $\mathcal{C}$ holds a $q$-ary $t$-design (at weight $w$).
	Index the codeword coordinates with an $n$-set $\mathcal{P}$. For a weight $w$ with $A_w \neq 0$, let $\mathcal{B}_w(\mathcal{C})$ denote the set of supports of the codewords in $\mathcal{A}_{w}(\mathcal{C})$. If the pair $(\mathcal{P}, \mathcal{B}_w(\mathcal{C}))$ forms a $t$-$(n, w, \lambda)$ design, we say that the code $\mathcal{C}$ supports a $t$-design (at weight $w$).
	Define the following strength parameters:
	\begin{align*}
		T_{w} = \max\{ t :  \mathcal{B}_{w}(\mathcal{C}) \text{ is a } t\text{-design}\}, \quad
		T_{w}^{(q)} = \max\{ t : \mathcal{A}_{w}(\mathcal{C})  \text{ is a $q$-ary } t\text{-design}\}.
	\end{align*}
	From Proposition \ref{qt->t} and Corollary \ref{2-ary-t-design}, it follows that $T_{w}\ge T_{w}^{(q)}$ and $T_{w}= T_{w}^{(2)} $.
	
	Recall a lemma that describes a property of codewords of certain weights in a linear code.
	
	\begin{Lemma}\cite[Lemma 4.25]{Ding} \label{repeat_q-1}
		Let $\mathcal{C}$ be an $[n,k,d]_{q}$ linear code. Let $h$ be the largest integer such that $h\le n$ and $h-\left \lfloor \frac{h+q-2}{q-1} \right \rfloor < d$. Then for any weight $w$ with $d\le w\le h$, every support of a codeword of weight $w$ is shared by exactly $q-1$ codewords in $\mathcal{A}_w(\mathcal{C})$.
	\end{Lemma}
	
	The following theorem, known as Assmus-Mattson theorem \cite[Theorem 4.2]{AM1969}, provides a fundamental connection between linear codes and ordinary designs.
	\begin{Theorem}\emph{(Assmus-Mattson)}\label{AM-Theorem}
		Let $\mathcal{C}$ be an $[n,k,d]_{q}$ linear code with dual distance $d^{\perp}$. Define $h$ as the largest integer satisfying $h\le n$ and $h-\left \lfloor \frac{h+q-2}{q-1} \right \rfloor < d$.
		Similarly, define $h^{\perp}$ using $d^{\perp}$. Denote the weight enumerator of $\cal C$ and $\mathcal{C^{\perp}}$ by $\sum_{i=0}^{n}A_{i}z^{i}$ and $\sum_{i=0}^{n}A^{\perp}_{i}z^{i}$, respectively.
		For a positive integer $t$ with $t<d$, let $m$ be the number of $i$ such that $A^{\perp}_{i}\ne 0$ for $1\le i\le n-t$. Let $t = d-m > 0$. Then
		\begin{itemize}
			\item [(1)] for every nonzero weight $w$ of $\cal C$ satisfying $d\le w\le h$, $\mathcal{B}_{w}(\mathcal{C})$ forms a $t$-design;
			\item [(2)] for every nonzero weight $w$ of $\mathcal{C^{\perp}}$ satisfying $d^{\perp}\le w\le \min\{n-t,h^{\perp}\}$,  $\mathcal{B}_{w}(\mathcal{C}^{\perp})$ forms a $t$-design.
		\end{itemize}
	\end{Theorem}
	The Assmus-Mattson theorem is a powerful tool for constructing $t$-designs from linear codes.
	This work aims to construct $q$-ary $t$-designs from linear codes.
	An analog of the Assmus-Mattson theorem will be established in the next section. The following is immediate.

	\begin{Proposition}\label{repeat}
		Let $\mathcal{C}$ be a linear code with parameters $[n,k,d]_{q}$, and let $h$ be the largest integer such that $h\le n$ and $h-\left \lfloor \frac{h+q-2}{q-1} \right \rfloor < d$.
		For any weight $w$ with $d\le w\le h$, if $\mathcal{A}_{w}(\mathcal{C})$ forms a $q$-ary $t$-$(n,w,\lambda)$ design, then $\mathcal{B}_{w}(\mathcal{C})$ forms a  $t$-$(n,w,(q-1)^{t-1}\lambda)$ design.
	\end{Proposition}
	\begin{proof}
		Assume that $d\le w\le h$ and that $\mathcal{A}_{w}(\mathcal{C})$ forms a $q$-ary $t$-$(n,w,\lambda)$ design.
		By Lemma \ref{repeat_q-1}, each support appearing in $\mathcal{B}_{w}(\mathcal{C})$ corresponds to exactly $q-1$ codewords of $\mathcal{A}_{w}(\mathcal{C})$.
		Combining this observation with Proposition \ref{qt->t} shows that $\mathcal{B}_{w}(\mathcal{C})$ forms a simple $t$-$(n,w,(q-1)^{t-1}\lambda)$ design.
	\end{proof}
	
	\subsection{Automorphism groups}
	The automorphism groups of linear codes play a significant role in establishing design properties.
	Let $\mathcal{C}$ be a linear code with parameters $[n,k,d]_{q}$. Denote by $\mathcal{P}$ the set of coordinate positions of codewords of $\mathcal{C}$, and let Sym$(\mathcal{P})$ be the symmetric group acting on $\mathcal{P}$. A codeword $\bm{c}\in \mathcal{C}$ can be written as $\bm{c}=(c_{x})_{x\in\mathcal{P}}$.
	The \emph{permutation automorphism group} $\PAut$$(\mathcal{C})$ is the subgroup of Sym$(\mathcal{P})$ that leaves $\mathcal{C}$ invariant. Explicitly, it is the subgroup consisting of all permutations $g\in$ Sym$(\mathcal{P})$ satisfying
	\[
	g(c_{x})_{x\in\mathcal{P}} = (c_{g^{-1}x})_{x\in\mathcal{P}} \in \mathcal{C} \text{ for all } (c_{x})_{x\in\mathcal{P}}\in \mathcal{C}.
	\]
	The \emph{monomial automorphism group} $\MAut$$(\mathcal{C})$ is the subgroup of the semidirect product $(\mathbb{F}_{q}^{*})^{n} \rtimes$ Sym$(\mathcal{P})$ that leaves $\mathcal{C}$ invariant. Specifically, it is the subgroup of Sym($\mathcal{P}$) consisting of all elements $((a_{x})_{x\in\mathcal{P}};g)$ satisfying
	\[
	((a_{x})_{x\in\mathcal{P}};g)(c_{x})_{x\in\mathcal{P}} = (a_{x}c_{g^{-1}x})_{x\in\mathcal{P}}\in \mathcal{C} \text{ for all } (c_{x})_{x\in\mathcal{P}} \in \mathcal{C}.
	\]
	Let Gal$(\mathbb{F}_{q})$ denote the Galois group of $\mathbb{F}_{q}$ over its prime field. The \emph{automorphism  group} $\Aut$$(\mathcal{C})$ is the subgroup of $(\mathbb{F}_{q}^{*})^{n} \rtimes$ (Sym$(\mathcal{P})\times$ Gal $(\mathbb{F}_{q})$) that maps $\mathcal{C}$ onto itself. Precisely, it is the subgroup consisting of all elements $((a_{x})_{x\in\mathcal{P}};g,\gamma)$ satisfying
	\[
	((a_{x})_{x\in\mathcal{P}};g,\gamma)(c_{x})_{x\in\mathcal{P}} = (a_{x}\gamma c_{g^{-1}x})_{x\in\mathcal{P}}\in \mathcal{C} \text{ for all } (c_{x})_{x\in\mathcal{P}}\in \mathcal{C}.
	\]
	The automorphism group $\Aut(\mathcal{C})$ is said to be $t$-\emph{transitive} (respectively, $t$-\emph{homogeneous}) if, for any two ordered $t$-tuples (respectively, $t$-subsets) of coordinate positions, there exists an element $((a_{x})_{x\in\mathcal{P}};g,\gamma)\in$ $\Aut(\mathcal{C})$ such that the permutation component $g$ maps the first tuple  (or set) to the second.
	It follows that if $\Aut(\mathcal{C})$ is $t$-transitive (or $t$-homogeneous), then it is also $i$-transitive (respectively $i$-homogeneous) for every positive integer $i\le t$. A $1$-transitive group is said to be transitive.
	
	Using the automophism group of a linear code, the following theorem gives a sufficient condition for the code to support ordinary designs.
	\begin{Theorem}\cite[Theorem 8.4.7]{fundamentals}\label{transitive-t-design}
		Let $\mathcal{C}\subseteq\mathbb{F}_{q}^{n}$ be a linear code whose automorphism group $\Aut(\mathcal{C})$ is $t$-transitive. Then every nonempty set $\mathcal{B}_{w}(\mathcal{C})$ with $w\ge t$ forms a $t$-design.
	\end{Theorem}
	\begin{Example}\label{transitive-no-q-design}
		As observed in \cite[Corollary 6.48, Example 6.50]{Ding}, the permutation automorphism group of the Dilix code $\overline{\Omega(3,3,2)}$ is doubly transitive. This code has parameters $[27,8,14]_{3}$ with weight enumerator
		$$1+810z^{14}+702z^{15}+1404z^{17}+780z^{18}+2106z^{20}+702z^{21}+54z^{26}+2z^{27}.$$
		$\mathcal{A}_{14}(\overline{\Omega(3,3,2)})$ does not form a ternary $2$-$(27,14,\lambda)$ design because the value of $\lambda$ obtained from
		$810 = \lambda \cdot 2^{2}{\binom{27}{2}}/{\binom{14}{2}}$
		is not an integer.
	\end{Example}
	Although Example \ref{transitive-no-q-design} shows that the $t$-transitivity of $\Aut(\mathcal{C})$ cannot be directly extended to guarantee $q$-ary $t$-designs, the  transitivity of $\Aut(\mathcal{C})$ is sufficient to produce $q$-ary $1$-designs, as shown in the following.
	\begin{Theorem}
		Let $\mathcal{C}\subseteq\mathbb{F}_{q}^{n}$ be a linear code whose automorphism group $\Aut$$(\mathcal{C})$ is transitive.
		Then any nonempty set $\mathcal{A}_{w}(\mathcal{C})$ forms a $q$-ary $1$-design.
	\end{Theorem}
	\begin{proof}
		Among all the vectors of weight 1 in $\mathbb{F}_{q}^{n}$, let $\bm{x}$ be one of such vectors that are covered by the maximum number of codewords of weight $w$. Suppose these distinct codewords are $\bm{c}_{1},\ldots,\bm{c}_{\lambda}$. Let $\bm{y}$ be any other vectors of weight 1 in $\mathbb{F}_{q}^{n}$. Since $\cal C$ is linear and $\Aut$$(\mathcal{C})$ is transitive, there exists an automorphism mapping $h = ((a_{i})_{i\in\mathcal{P}};g,\gamma)\in$ $\Aut$$(\mathcal{C})$ and an element $\alpha$ of $\mathbb{F}_{q}^{*}$ such that $\alpha h(\bm{x}) = \bm{y}$. Then $\alpha h(\bm{c}_{1}),\ldots,\alpha h(\bm{c}_{\lambda})$ are distinct codewords of $\mathcal{C}$ that cover $\bm{y}$. By the maximality of $\lambda$, there are no more codewords that cover $\bm{y}$. Hence the codewords of any weight $w$ of $\mathcal{C}$ form a $q$-ary $1$-$(n,w,\lambda)$ design.
	\end{proof}
	
	We proceed in the next section to establish a general framework for constructing $q$-ary $t$-designs from linear codes.
	For convenience, we summarize some notations and definitions adopted throughout the paper in Table \ref{table-all}.
	
	\begin{table}[ht]
		\caption{Table of definitions and notations}\vspace{0.2cm}	
		\footnotesize
		\centering
		\renewcommand\arraystretch{1.1}
		\begin{tabular}{ccc}
			\toprule
			Notation & Meaning & Remarks  \\
			\midrule
			$\mathbb{F}_{q}$ & Finite field with $q$ elements & Sec.2  \\
			$\mathbb{F}_{q}^{*}$ & $\mathbb{F}_{q}\setminus\{\bm{0}\}$ & Sec.2 \\
			$\mathbb{F}_{q}^{n}$ & $n$-Dimensional vector space over $\mathbb{F}_{q}$ & Sec.2 \\
			$\binom{\mathcal{P}}{k}$ & Set of all $k$-subsets of $\mathcal{P}$ & Sec.2 \\
			$[n]$ & The set $\{1,2,\dots,n\}$ & Sec.2 \\
			$\mathrm{wt}$($\bm{x}$) & Hamming weight of vector $\bm{x}$ & Sec.2  \\
			$\mathrm{supp}(\bm{x})$ & Support of vector $\bm{x}$ & Sec.2 \\
			$W(\cal C)$ & Set of distinct nonzero weights in a code $\cal C$  & Sec.2  \\
			$d=d(\mathcal{C})$ &Minimum distance of $\mathcal{C}$ & Sec.2 \\
			$d^{\perp}=d(\mathcal{C}^{\perp})$ & Dual distance of $\mathcal{C}$ & Sec.2  \\
			$s=s(\mathcal{C})$ & Number of nonzero weights in $\cal C$ & Sec.2 \\
			$s^{\perp}=s(\mathcal{C}^{\perp})$ & External distance of $\mathcal{C}$ & Sec.2 \\
			$e=e(\mathcal{C})$ & Packing radius of $\mathcal{C}$ & Sec.2 \\
			$\rho=\rho(\mathcal{C})$ & Covering radius of $\mathcal{C}$ & Sec.2 \\
			$\mathcal{C}^{\{m\}}$ & Punctured code of $\mathcal{C}$ on the $m$-th coordinate  & Sec.2  \\
			$\mathcal{C}_{\{m\}}$ & Shortened code of $\mathcal{C}$ on the $m$-th coordinate  & Sec.2  \\
			$A_{w}$ & Number of codewords of weight $w$ in $\cal C$ & Sec.2 \\
			$\mathcal{A}_{w}(\mathcal{C})$ & Set of all codewords  of weight $w$ in $\cal C$ & Sec.2 \\
			$\mathcal{B}_{w}(\mathcal{C})$ & Set of supports of all codewords in $\mathcal{A}_{w}(\mathcal{C})$ & Sec.2 \\
			$B_{\bm{x},i}$ & Number of codewords at distance $i$ from $\bm{x}$ & Sec.3 \\
			$B_{\bm{x},i}(w)$ & Number of weight-$w$ codewords in $B_{\bm{x},i}$ & Sec.3 \\
			$\sigma_{l}$ & Elementary symmetric polynomial of degree $l$ & Sec.5 \\
			$U$ & Subgroup of $\mathbb{F}_{q^{2}}^{*}$ comprised of all $(q+1)$-th roots of unity & Sec.5 \\		
			\bottomrule
		\end{tabular}
		\label{table-all}
	\end{table}

	\section{General framework of constructing $q$-ary designs from linear codes}
	\setcounter{equation}{0}
	This section establishes two principal criteria for obtaining $q$-ary $t$-designs from linear codes: a Standard Criterion and a Puncturing-Shortening Criterion. Additionally, we establish an equivalence between $t$-regular codes and $q$-ary $t$-designs under certain conditions.
	
	\subsection{Standard Criterion}
	
	Delsarte \cite{dsd's'} introduced the concept of a $q$-ary $t$-design and showed that a $q$-ary linear code can yield such designs when its dual distance $d^{\perp}$ and the number $s$ of its nonzero weights satisfy $d^{\perp} > s$.
	In \cite{Delsarte1973}, regularity properties of a code were linked to its minimum distance $d$ and external distance $s^{\perp}$.
	Later, Goethals and van Tilborg \cite{Tilborg1975,UP} noted that a $t$-regular code produces a $q$-ary $t$-design when $d \ge 2t$. Building on these works, we develop these ideas into a Standard Criterion for identifying $q$-ary designs from linear codes.
	
	\begin{Lemma}\label{d'-s}\cite[Theorem 5.3]{dsd's'}
		Let $\mathcal{C}$ be a $q$-ary code with dual distance $d^{\perp}$ and let $s$ be the number of its nonzero weights. If $t = d^{\perp}-s > 0$, then the codewords of any fixed nonzero weight in $\mathcal{C}$ form a $q$-ary $t$-design.
	\end{Lemma}
	
	For a code $\mathcal{C} \subseteq \mathbb{F}_{q}^{n}$ and any vector $\bm{x} \in \mathbb{F}_{q}^{n} $, let $B_{\bm{x},i}$ denote the number of codewords at distance $i$ from $\bm{x}$. The \emph{outer distribution matrix} of $\mathcal{C}$ is a $q^{n} \times (n+1)$ matrix $B$ with entries
	\begin{equation}\label{Bxi}
		B_{\bm{x},i} = |\{\bm{v} \in \mathcal{C} : d(\bm{x},\bm{v})=i \}|.
	\end{equation}
	We further denote by $B_{\bm{x},i}(w)$ the number of codewords of weight $w$ at distance $i$ from $\bm{x}$.
	
	\begin{Definition}
		Let $0\le t\le \rho$. The code $\mathcal{C}$ with covering radius $\rho$ is \emph{$t$-regular} if for all $i =0,1,\dots,\rho$ the value $B_{\bm{x},i}$ depends only on $i$ and $d(\bm{x},\mathcal{C})$ for all $\bm{x}$ with $d(\bm{x},\mathcal{C}) \le t$. The code is \emph{completely regular} if it is $\rho$-regular.
	\end{Definition}
	
	%d-s regular
	\begin{Lemma}\cite[Theorems 5.11-5.13]{Delsarte1973}\label{d-s regular}
		Let $\mathcal{C}$ be a $q$-ary code with minimum distance $d$ and external distance $s^{\perp}$.
		\begin{itemize}
			\item [(1)] If $s^{\perp} \le d \le 2s^{\perp}-2$, the code $\mathcal{C}$ is $(d-s^{\perp})$-regular;
			\item [(2)] If $d$ takes one of the values $2s^{\perp}-1$, $2s^{\perp}$ or $2s^{\perp}+1$, then $\mathcal{C}$ is completely regular with covering radius $\rho = s^{\perp}$.
		\end{itemize}
	\end{Lemma}
	%	In particular, every perfect code is completely regular.
	The significance of $t$-regularity lies in its powerful combinatorial consequence as follows.
	%t-regular  -->  t-design
	\begin{Lemma}
		\label{d-s'}\cite[Theorem 2.4.7]{UP} Let $\mathcal{C}$ be a $t$-regular code with minimum distance $d \ge 2t$. Then every nonempty set $\mathcal{A}_{w}(\mathcal{C})$ is a $q$-ary $t$-design.
	\end{Lemma}
	Observe that if $t = d-s^{\perp}$ in Lemma \ref{d-s'}, the condition $d\ge2t$ becomes redundant.
	\begin{Corollary} \label{d-s'*}
		Let $\mathcal{C}$ be a $q$-ary linear code with minimum distance $d$ and external distance $s^{\perp}$.
		If $d>s^{\perp}$, then every nonempty set $\mathcal{A}_{w}(\mathcal{C})$ is a $q$-ary $(d-s^{\perp})$ design.
	\end{Corollary}
	\begin{proof}
		Let $\mathcal{C}$ be an $[n,k,d]_{q}$ code with packing radius $e=\lfloor\frac{d-1}{2}\rfloor$. It is known that $s^{\perp}\ge \lfloor\frac{d-1}{2}\rfloor $, see \eqref{e<rho<s'}.
		If $s^{\perp}=\lfloor\frac{d-1}{2}\rfloor$, then  $\mathcal{C}$ is a perfect code and it has odd minimum distance, see, for example, \cite[Section 10.4]{Ding}. Then $d = 2s^{\perp}+1$ and consequently $e=\rho=s^{\perp}$.
		Applying \cite[Theorem 2]{CR}, in a perfect code, every nonempty set $\mathcal{A}_{w}(\mathcal{C})$ forms a $q$-ary $(e+1)$ design, i.e., a $q$-ary $(d-s^{\perp})$ design, see also Theorem \ref{perfect-q(e+1)} herein.
		If $s^{\perp}>\lfloor\frac{d-1}{2}\rfloor $, then $d\ge2(d-s^{\perp})$.
		In this case, $\cal C$ is $(d-s^{\perp})$-regular by Lemma \ref{d-s regular}. Combining with Lemma \ref{d-s'} yields that each nonempty set $\mathcal{A}_{w}(\mathcal{C})$ is a $q$-ary $(d-s^{\perp})$ design.
	\end{proof}
	
	We combine Lemma \ref{d'-s} with Corollary \ref{d-s'*} to generate a main approach of identifying $q$-ary designs from linear codes, which could be viewed as an analog of the Assmus-Mattson theorem.
	\begin{Theorem}\label{Standard-Criterion}
		\emph{(Standard Criterion)} Let $\mathcal{C}$ be a $q$-ary linear code with minimum distance $d$, $s$ nonzero weights,  dual distance $d^{\perp}$ and external distance $s^{\perp}$.
		Assume that $d>s^{\perp}$ or $d^{\perp}>s$. Define $t = \max\{d-s^{\perp}, d^{\perp}-s\}$.
		Then the following hold:
		\begin{enumerate}
			\item [(1)] For any nonzero weight $w$ of $\cal C$, every nonempty set $\mathcal{A}_{w} (\cal C)$ forms a $q$-ary $t$-design.
			\item [(2)] For any nonzero weight $w$ of $\mathcal{C^{\perp}}$, every nonempty set $\mathcal{A}_{w} (\mathcal{C^{\perp}})$ forms a $q$-ary $t$-design.
		\end{enumerate}
	\end{Theorem}
	
	\subsection{Puncturing-Shortening Criterion}
	Puncturing and shortening are fundamental operations in coding theory for constructing new codes from existing ones.
	Assmus, Goethals, and Mattson studied properties of $q$-ary $t$-designs in detail and examined their relationship with codes and punctured codes \cite{punctured-t}.
	This subsection extends this idea to establish a Puncturing-Shortening Criterion.
	
	\begin{Lemma}\label{DM0}\cite[Corollary 2.4]{punctured-t}
		Let $\mathcal{A}$ be a $q$-ary $t$-$(n,w,\lambda)$ design with $t\ge 1$. For any coordinate position m, and any $\alpha \in \mathbb{F}^{*}_{q}$, the sets
		\[
		D_{m}(\alpha) = \{ \bm{u} \in \mathcal{A} : u_m = \alpha \} \quad  \text{and} \quad D_{m}(0) = \{ \bm{u} \in \mathcal{A} : u_m = 0 \}
		\]
		respectively yield a $q$-ary $(t-1)$-$(n-1,w-1,\lambda)$ design and a $q$-ary $(t-1)$-$(n-1,w,\eta)$ design, where $\eta = \lambda (q-1)(n-w)/(w-t+1)$.
	\end{Lemma}

	Applying Lemma \ref{DM0} to a linear code $\mathcal{C}$, we obtain the following: if $\mathcal{A}_{w}(\mathcal{C})$ is a $q$-ary $t$-design,
	then the shortened code $\mathcal{C}_{\{m\}}$ holds a $q$-ary $(t-1)$-design.
	A natural question is under what conditions the strength $t$ of the $q$-ary design held by the code is still retained for the punctured code or the shortened code.
	A known result \cite[Theorem 4.4]{punctured-t} establishes such conditions for puncturing.
	By developing this with the structural insight from the Standard Criterion and Lemma \ref{DM0}, we obtain the following.
	
	\begin{Theorem}\label{punctured} \emph{(Puncturing-Shortening Criterion)}
		Let $\mathcal{C}$ be an $[n,k,d]_{q}$ linear code with external distance $s^{\perp}< d$.
		Suppose that $n\in W(\mathcal{C}^{\perp})$.
		Define $t=d-s^{\perp}$ and $W'=\{w\in W(\mathcal{C}) : w-1\notin W(\mathcal{C})\}$.
		Then the following hold:
		\begin{itemize}
			\item[(1)] For every nonzero weight $w\in W(\mathcal{C})$, $\mathcal{A}_{w}(\mathcal{C})$ is a $q$-ary $t$-$(n,w,\lambda)$ design for some integer $\lambda$.
			\item[(2)] For the punctured code $\mathcal{C}^{\{m\}}$,  $\mathcal{A}_{w-1}(\mathcal{C}^{\{m\}})$ is a $q$-ary $t$-$(n-1,w-1,\mu)$ design for every nonzero weight $w\in W(\mathcal{C})$. Specifically, if $w\in W'$, then $\mu = \lambda\frac{w-t}{n-t}$.
			\item[(3)] For the shortened code $\mathcal{C}_{\{m\}}$, $\mathcal{A}_{w}(\mathcal{C}_{\{m\}})$ is a $q$-ary $t$-$(n-1,w,\lambda\frac{n-w}{n-t})$ design for every nonzero weight $w\in W'$.
		\end{itemize}
	\end{Theorem}
	\begin{proof}
		Part $(1)$ follows directly from the Standard Criterion.
		For $(2)$, the dual code $(\mathcal{C}^{\{m\}})^{\perp}$ is obtained from the subcode $\{\bm{u}\in\mathcal{C^{\perp}} : u_{m}=0\}$  by deleting the $m$-th coordinate. Consequently, $d(\mathcal{C}^{\{m\}})=d-1$, and $W((\mathcal{C}^{\{m\}})^{\perp})\subseteq W(\mathcal{C^{\perp}})\setminus\{n\}$, so $s^{\perp}(\mathcal{C}^{\{m\}})\le s^{\perp}-1$ by noting $n\in W(\mathcal{C^{\perp}})$.
		Applying the Standard Criterion shows that $\mathcal{A}_{w-1}(\mathcal{C}^{\{m\}})$ is a $q$-ary $t$-$(n-1,w-1,\mu)$ design for some $\mu$.
		When $w\in W'$, comparing the block counts of
		$\{\bm{c}\in\mathcal{A}_{w}(\mathcal{C}) : c_{m}\in\mathbb{F}_{q}^{*}\}$
		and $\mathcal{A}_{w-1}(\mathcal{C}^{\{m\}})$ via Lemma \ref{DM0} yields
		\[
		(q-1)\frac{\lambda\binom{n-1}{t-1}(q-1)^{t-1}}{\binom{w-1}{t-1}} = \dfrac{\mu\binom{n-1}{t}(q-1)^{t}}{\binom{w-1}{t}}.
		\]
		Then we have $\mu = \lambda\frac{w-t}{n-t}$.
		
		For $(3)$, let $w\in W'$.  $\mathcal{A}_{w}(\mathcal{C})$ is a $q$-ary $t$-$(n,w,\lambda)$ design by $(1)$. $\mathcal{A}_{w-1}(\mathcal{C}^{\{m\}})$ is a $q$-ary $t$-$(n-1,w-1,\mu)$ design with $\mu = \frac{\lambda (w-t)}{n-t}$ by $(2)$.
		The set of codewords of weight $w$ in $\mathcal{C}$ can be partitioned into two disjoint subsets: those with a zero at coordinate $m$, which correspond to the codewords of weight $w$ in the shortened code $\mathcal{C}_{\{m\}}$, and those with a nonzero entry at coordinate $m$, which project onto codewords of weight $w-1$ in $\mathcal{C}^{\{m\}}$. Therefore, the codewords of weight $w$ in $\mathcal{C}_{\{m\}}$ is a $q$-ary $t$-$(n-1,w,\lambda-\mu)$ design.
		Substituting $\mu = \lambda\frac{w-t}{n-t}$ gives $ \lambda-\mu = \frac{\lambda (n-w)}{n-t}$, which completes the proof.
	\end{proof}
	
	\begin{Corollary}
		Let $\mathcal{C}$ be a $q$-ary $[n,k,d]$ linear code with external distance $s^{\perp}< d$. Define $t=d-s^{\perp}$. Then the following statements hold:
		\begin{itemize}
			\item[(1)] If $n\in W(\mathcal{C}^{\perp})$, $\mathcal{A}_{d}(\mathcal{C}_{\{m\}})$ is a $q$-ary $t$-design.
			\item[(2)] If $n\in W(\mathcal{C}^{\perp})$ and $\mathcal{C}$ is a divisible code, then $\mathcal{A}_{w}(\mathcal{C}_{\{m\}})$ is a $q$-ary $t$-design for every nonzero weight $w\in W(\mathcal{C})$.
			\item[(3)] If $\mathcal{C}$ is an even-like code (with all codewords $\bm{c}\in\mathcal{C}$ satisfying $\sum_{i=1}^{n}c_{i}=0$), then the punctured code $\mathcal{C}^{\{m\}}$ holds a $q$-ary $t$-design at any nonzero weight $w$.
		\end{itemize}
	\end{Corollary}
	\begin{proof}
		The statements of (1) and (2) are obvious by Theorem $\ref{punctured}$ as the minimum distance $d\in W'$ and $W(\mathcal{C}) = W'$ for a divisible code. For (3),
		if $\mathcal{C}$ is even-like, then the all-one vector $(1,\dots,1)$ belongs to $\mathcal{C}^{\perp}$, implying $n \in W(\mathcal{C}^{\perp})$. The statement is then a direct consequence of Theorem~\ref{punctured}.
	\end{proof}

	\subsection{$q$-Ary $t$-designs and $t$-regular codes}
	We now turn to the deeper structural relationship in a code between $t$-regularity and holding $q$-ary $t$-designs.
	To facilitate further analysis, we introduce a lemma that exploits the properties of $q$-ary $t$-designs.
	\begin{Lemma}
		\label{l-xyz}
		Let $H$ be the $n\times |\mathcal{A}|$ matrix whose columns are all blocks of a $q$-ary $t$-$(n, w, \lambda)$ design $\mathcal{A}$. For any non-negative integers $x,y,z$ satisfying $x+y+z\le t$, let $H'$ be any $(x+y+z)\times |\mathcal{A}|$ submatrix of $H$ and let $\alpha$ be a fixed $(x+y+z)$-tuple with all entries nonzero. Partition the coordinates of $\alpha$ into three disjoint sets $X,Y,Z$ of sizes $x,y,z$ respectively.
		Among the columns of $H'$, denote by $\lambda_{x,y,z}(\alpha)$ the number of blocks that coincide with $\alpha$ on $X$, differ from $\alpha$ and are nonzero on $Y$, and are identically zeros on $Z$. Then $\lambda_{x,y,z}(\alpha)$ is a constant $\lambda_{x,y,z}$ depending only on $x,y,z$ and $\lambda$.
	\end{Lemma}
	\begin{proof}
		First, observe that $\lambda_{0,0,0}(\alpha)=|\mathcal{A}|$. When $z=0$, the condition on $Z$ is void, and each of the $y$ positions in $Y$ must be a nonzero element different from the corresponding entry of $\alpha$, giving $q-2$ choices. Hence $\lambda_{x,y,0}(\alpha)=(q-2)^{y}\lambda_{x+y}$ is independent of the choice of $\alpha$ where $\lambda_{i}$ is defined in \eqref{q-lambda-i}.
		Now assume that $z>0$. Fix a position $m\in Z$.
		Then let $\alpha'$ be the $(x+y+z-1)$-tuple obtained from $\alpha$ by suppressing $\alpha_{m}$, and let $\alpha^{(1)},\alpha^{(2)},\dots,\alpha^{(q-1)}$ be the $q-1$ distinct $(x+y+z)$-tuple obtained from $\alpha$ by replacing $\alpha_{m}$ by each of the $q-1$ nonzero elements of $\mathbb{F}_{q}$. Then we have
		$$\lambda_{x,y,z}(\alpha)=\lambda_{x,y,z-1}(\alpha')-\sum_{i=1}^{q-1}\lambda_{x+1,y,z-1}(\alpha^{(i)}).$$
		Combining this recurrence with the initial condition $\lambda_{x,y,0}(\alpha)=(q-2)^{y}\lambda_{x+y} = \lambda(q-2)^{y}(q-1)^{t-x-y}\binom{n-x-y}{t-x-y}\big/\binom{w-x-y}{t-x-y}$, we obtain	
		$$\lambda_{x,y,z}(\alpha) = \lambda(q-2)^{y}(q-1)^{t-x-y}\binom{n-x-y-z}{w-x-y}\bigg/\binom{n-t}{w-t},$$
		which completes the proof.
	\end{proof}
	
	While Lemma \ref{d-s'} establishes that a $t$-regular code may hold $q$-ary $t$-designs, it is natural to ask whether the converse holds. Using Lemma \ref{l-xyz}, we prove that this is indeed the case, leading to the following theorem.	
	\begin{Theorem}
		Let $t$ be a positive integer and let $\mathcal{C}$ be a linear $[n,k,d]_{q}$ code whose minimum distance satisfies $d\ge 2t$. Then $\mathcal{C}$ is $t$-regular if and only if, for every nonzero weight $w$, $\mathcal{A}_{w}(\mathcal{C})$ is a $q$-ary $t$-design.
	\end{Theorem}
	\begin{proof}
		The necessity follows from Lemma \ref{d-s'}.
		For sufficiency, assume that for each nonzero weight $w$, $\mathcal{A}_{w}(\mathcal{C})$ is a $q$-ary $t$-$(n,w,\lambda^{(w)}_{t})$ design.
		By Lemma \ref{q-design-li}, $\mathcal{A}_{w}(\mathcal{C})$ is also a $q$-ary $i$-$(n,w,\lambda^{(w)}_{i})$ design for every $1\le i\le t$.
		Fix a vector $\bm{x}\in \mathbb{F}_{q}^{n}$ with wt$(\bm{x})\le t$. Since $d\ge 2t$, we have $d(\bm{x},\mathcal{C})=d(\bm{x},\bm{0})=$ wt$(\bm{x})$.
		We analyze the numbers $B_{\bm{x},j}$ denoted in $\eqref{Bxi}$ for $j=0,1,\dots,n$.
		
		%		\noindent\textbf{Case 1:}
		If $j\in\{0,1,\dots,d-\mathrm{wt}(\bm{x})-1\}\setminus\{\mathrm{wt}(\bm{x})\}$, then $B_{\bm{x},j}=0$. Moreover, $B_{\bm{x},\mathrm{wt}(\bm{x})}=1$ except possibly when $d=2t$ and wt$(\bm{x})\ne t$.
		
		Now consider the value $B_{\bm{x},j}$ with $j=d-\mathrm{wt}(\bm{x})$. Any codeword $\bm{c}\in \mathcal{C}$ with $d(\bm{x},\bm{c})=j$ must satisfy wt$(\bm{c})\le d$ by the triangle inequality;  hence either $\bm{c}=\bm{0}$ or wt$(\bm{c})=d$.
		If wt$(\bm{c})=d$, then $\bm{c}$ covers $\bm{x}$.
		Therefore, the number of such nonzero codewords is exactly $\lambda^{(d)}_{\mathrm{wt}(\bm{x})}$.
		When $d=2t$ and wt$(\bm{x})=t$, the zero codeword  also lies at distance $j$ from $\bm{x}$. Consequently,
		\[
		B_{\bm{x},j}=
		\begin{cases}
			\lambda^{(d)}_{\mathrm{wt}(\bm{x})}+1, & \text{if } d=2t \text{ and } \mathrm{wt}(\bm{x})=t, \\[4pt]
			\lambda^{(d)}_{\mathrm{wt}(\bm{x})}, & \text{otherwise}.
		\end{cases}
		\]
		In either situation $B_{\bm{x},j}$ is a constant depending only on $d(\bm{x},\mathcal{C})$.
		
		Finally suppose $j\ge d-\mathrm{wt}(\bm{x})$.
		We decompose $B_{\bm{x},j}$ according to the weight of the codewords:
		\begin{align}\label{eq-BXJ}
			B_{\bm{x},j} = B_{\bm{x},j}(\mathrm{wt}(\bm{x})+j) \;+\; \sum_{z=d}^{\mathrm{wt}(\bm{x})+j-1}B_{\bm{x},j}(z),
		\end{align}
		The first term equals $\lambda^{(\mathrm{wt}(\bm{x})+j)}_{\mathrm{wt}(\bm{x})}$, because a codeword of weight $\mathrm{wt}(\bm{x})+j$ that is at distance $j$ from $\bm{x}$ must cover $\bm{x}$.
		For a codeword $\bm{c}$ of weight $z$ with $d\le z\le \mathrm{wt}(\bm{x})+j-1$ and $d(\bm{x},\bm{c}) = j$, note that
		\begin{align}
			d(\bm{x},\bm{c}) = \mathrm{wt}(\bm{x})+\mathrm{wt}(\bm{c})-|\mathrm{supp}(\bm{x})\cap \mathrm{supp}(\bm{c})| - |\{i\in[n] : {x}_{i}={c}_{i}\ne 0\}|.
		\end{align}
		Let $s = |\mathrm{supp}(\bm{x})\cap \mathrm{supp}(\bm{c})|$ and $m = |\{i\in[n] : {x}_{i}={c}_{i}\ne 0\}|$.
		Then $j=z+\mathrm{wt}(\bm{x})-s-m$.
		Applying Lemma \ref{l-xyz} to the $q$-ary $t$-design $\mathcal{A}_{z}(\mathcal{C})$ (noting that $m+(s-m)+(\mathrm{wt}(\bm{x})-s) = \mathrm{wt}(\bm{x})\le t$), we obtain
		\begin{equation}\label{sumBxj}
			\sum_{z = d}^{\mathrm{wt}(\bm{x})+j-1}B_{\bm{x},j}(z) =
			\sum_{z=d}^{\mathrm{wt}(\bm{x})+j-1}\sum_{0\le m\le s\le \mathrm{wt}(\bm{x})\atop z+\mathrm{wt}(\bm{x})-s-m=j}\lambda^{(z)}_{m,s-m,\mathrm{wt}(\bm{x})-s},
		\end{equation}
		where each $\lambda^{(z)}_{m,s-m,\mathrm{wt}(\bm{x})-s}$ is a constant depending only on $m$, $s$, $j$ and $d(\bm{x},\mathcal{C})$.
		
		Combining the above arguments, we see that $B_{\bm{x},j}$ depends only on $\mathrm{wt}(\bm{x})$ and $j$. Hence $\mathcal{C}$ is $t$-regular.
	\end{proof}

	\section{Applications of general framework}
	\setcounter{equation}{0}
	This section presents specific families of linear codes that hold $q$-ary $t$-designs through the general framework established in Section $3$. We focus on five prominent classes: MDS codes, perfect codes, one-weight and two-weight codes, and extremal self-dual codes.
	
	\subsection{MDS codes and perfect codes}
	This subsection presents equivalent characterizations of linear MDS codes and linear perfect codes in the framework of $q$-ary designs.
	
	\begin{Lemma}\cite{AM1969}\label{MDS-1}
		A linear code is $\mathrm{MDS}$ if and only if the supports of its minimum weight codewords form a complete design.
	\end{Lemma}
	This equivalence can also be expressed in the language of $q$-ary design as follows.
	\begin{Theorem}\label{MDS-Q1}
		Let $\cal C$ be an $[n,k,d]_{q}$ linear code. If $\cal C$ is $\mathrm{MDS}$, then for every nonzero weight $w$, $\mathcal{A}_{w}(\mathcal{C})$ forms a $q$-ary $1$-design. Moreover, $\cal C$ is $\mathrm{MDS}$ if and only if $\mathcal{A}_{d}(\mathcal{C})$ forms a $q$-ary $1$-$\left (n,d,\binom{n-1}{d-1}\right )$ design.
	\end{Theorem}
	\begin{proof}
		Assume $\cal C$ is a linear $\mathrm{MDS}$ code with minimum distance $d = n-k+1$. Its dual code $\mathcal{C}^{\perp}$ is also MDS with parameters $[n,n-k,k+1]_{q}$. Hence, the external distance $s^{\perp}$ of $\cal C$ satisfies $s^{\perp}\le n-(k+1)+1 = n-k$ which gives $d-s^{\perp}\ge 1$.
		By the Standard Criterion, for every nonzero weight $w$, $\mathcal{A}_{w}(\mathcal{C})$ forms a $q$-ary $1$-design.
		Furthermore, by the fact that the number of codewords of weight $d$ in $\cal C$ equals $(q-1)\binom{n}{d}$, we conclude that  $\mathcal{A}_{d}(\mathcal{C})$ forms a $q$-ary $1$-$(n,d,\binom{n-1}{d-1})$ design.
		
		Conversely, suppose that $\mathcal{A}_{d}(\mathcal{C})$ is a $q$-ary $1$-$(n,d,\binom{n-1}{d-1})$ design. Then the number of codewords of weight $d$ in $\cal C$ is  $(q-1)\binom{n}{d}$, implying from Lemma \ref{repeat_q-1} that every $d$-subset of $[n]$ appears as the support of a codeword. Consequently, $\mathcal{B}_{d}(\mathcal{C})$ forms a complete design.
		Applying Lemma \ref{MDS-1}, we conclude that $\cal C$ is $\mathrm{MDS}$.
	\end{proof}
	
	While the above theorem guarantees $q$-ary $1$-designs, certain MDS codes can yield $q$-ary $t$-designs with strength $t>1$, such as the doubly-extended Reed-Solomon codes (see Theorem \ref{DRS-q2}). Conversely, some MDS codes yield only $q$-ary $1$-designs, as shown by the following example.
	\begin{Example}\label{Example-RS}
		Consider the Reed-Solomon code $\cal C$ with parameters $[15,4,12]_{16}$, which is an $\mathrm{MDS}$ code. In this case, the set $\mathcal{A}_{12}(\mathcal{C})$ forms only a $16$-ary $1$-$(15,12,364)$ design, i.e., $T_{d}^{(q)} = 1$.
	\end{Example}
	
	The following lemma and theorem establish the fundamental connections between perfect codes and combinatorial designs, first in the classical and then in the $q$-ary setting.
	\begin{Lemma}\cite[Theorem 3.1]{AM1974}\label{perfect-t}
		A linear code $\mathcal{C}\subseteq \mathbb{F}_{q}^{n}$ with minimum distance $d=2e+1$ is perfect if and only if  $\mathcal{B}_{d}(\mathcal{C})$ forms an $(e+1)$-$(n,d,(q-1)^{e})$ design.
	\end{Lemma}
	\begin{Theorem}\label{perfect-q(e+1)}
		A linear code $\mathcal{C}\subseteq \mathbb{F}_{q}^{n}$ with minimum distance $d=2e+1$ is perfect if and only if $\mathcal{A}_{d}(\mathcal{C})$ forms a $q$-ary $(e+1)$-$(n,2e+1,1)$ design.
	\end{Theorem}
	\begin{proof}
		Suppose that $\mathcal{C}\subseteq \mathbb{F}_{q}^{n}$ with minimum distance $d=2e+1$ is a linear perfect code. For any vector $\bm{x}\in \mathbb{F}_{q}^{n}$ of weight $e+1$, there exists exactly one codeword $\bm{c}\in \mathcal{C}$ at distance $e$ from $\bm{x}$, and wt$(\bm{c})=2e+1$. Hence, $\bm{c}$ covers $\bm{x}$, and $\mathcal{A}_{d}(\mathcal{C})$ is a $q$-ary $(e+1)$-$(n,2e+1,1)$ design.
		Conversely, assume that $\mathcal{A}_{d}(\mathcal{C})$ forms a $q$-ary $(e+1)$-$(n,2e+1,1)$ design.
		By Proposition $\ref{repeat}$, after disregarding scalar multiples, $\mathcal{B}_{d}(\mathcal{C})$ form an $(e+1)$-$(n,2e+1,(q-1)^{e})$ design. By Lemma \ref{perfect-t}, $\mathcal{C}$ is perfect.
	\end{proof}
	More generally, the properties of linear perfect codes guarantee that all their nonzero weight codewords yield $q$-ary $(e+1)$-designs.
	\begin{Theorem}
		\label{perfect-design}
		Let $\mathcal{C}$ be a $q$-ary linear perfect code of length $n$ with minimum distance $d=2e+1$. Then for every nonzero weight $w$, $\mathcal{A}_{w}(\mathcal{C})$ is a $q$-ary $(e+1)$-design.
	\end{Theorem}
	\begin{proof}
		Let $\mathcal{C}$ be an $[n,k,d]_{q}$ linear perfect code. Then $\mathcal{C}$ is completely regular with covering radius $\rho=e$ by \cite[Corollary 10.15]{Ding}.
		Take a vector $\bm{x}\in\mathbb{F}_{q}^{n}$ of weight $e+1$.
		Then $d(\bm{x},\mathcal{C})=e$, and by the complete regularity of $\cal C$, the numbers $B_{\bm{x},i}$ (see \eqref{Bxi}) depends only on $i$ and the parameter $e$.
		First consider the case $w=d=2e+1$. By Theorem \ref{perfect-q(e+1)}, $\mathcal{A}_{d}(\mathcal{C})$ is a $q$-ary $(e+1)$-$(n,d,1)$ design.
		Now assume by induction that for weight $k$ with $k<w$ and $\mathcal{A}_{k}(\mathcal{C})\ne\emptyset$, the set $\mathcal{A}_{k}(\mathcal{C})$ is a $q$-ary $(e+1)$-$(n,k,\lambda^{(k)})$ design. We prove the statement for weight $w$.
		The number of codewords of weight $w$ that cover $\bm{x}$ is given by $B_{\bm{x},w-\mathrm{wt}(\bm{x})}(w)$. Analogous to \eqref{sumBxj}, we obtain
		\begin{align*}
			B_{\bm{x},w-\mathrm{wt}(\bm{x})}(w) &= B_{\bm{x},w-\mathrm{wt}(\bm{x})}-\sum_{z=0}^{w-1}B_{\bm{x},w-\mathrm{wt}(\bm{x})}(z) \\
			&= B_{\bm{x},w-\mathrm{wt}(\bm{x})}-\sum_{z=0}^{w-1} \sum_{0\le m\le s\le \mathrm{wt}(\bm{x})\atop z+\mathrm{wt}(\bm{x})-s-m=w-\mathrm{wt}(\bm{x})}\lambda^{(z)}_{m,s-m,\mathrm{wt}(\bm{x})-s}.
		\end{align*}
		where each constant $\lambda^{(z)}_{m,s-m,\mathrm{wt}(\bm{x})-s}$ is determined by Lemma \ref{l-xyz}. Because $B_{\bm{x},w-\mathrm{wt}(\bm{x})}$ is a constant and the inner sums depend only on wt$(\bm{x})$, the expression on the right-hand side is independent of the particular choice of $\bm{x}$. Therefore, the number of codewords of weight $w$ covering any fixed vector of weight $e+1$ is a constant. This implies that $\mathcal{A}_{w}(\mathcal{C})$ is a $q$-ary $(e+1)$-$(n,w,\lambda)$ design for some $\lambda$. By induction, the theorem holds for all admissible weights $w$.
	\end{proof}
	
	By Theorem \ref{perfect-design}, the codewords of a fixed weight in linear perfect codes yield $q$-ary $t$-designs. Specific instances are provided below.
	\begin{Example}
		Let $\mathcal{C}$ be an $[n,n-m,3]_{q}$ Hamming code of length $n = \frac{q^{m}-1}{q-1}$, where $m\ge2$. Then every nonempty set $\mathcal{A}_{w}(\mathcal{C})$ forms a $q$-ary $2$-design. In particular, $\mathcal{A}_{3}(\mathcal{C})$ is a $q$-ary $2$-$(n,3,1)$ design and $\mathcal{B}_{3}(\mathcal{C})$ is a $2$-$(n,3,q-1)$ design. In this case, $T_{3}=T_{3}^{(q)} = 2$.
	\end{Example}
	\begin{Example}\label{Ex-golay}
		Let $\cal C$ be the $[11,6,5]_{3}$ ternary Golay code. Then every nonempty set $\mathcal{A}_{w}(\mathcal{C})$ is a ternary $3$-design as listed in the first row of Table \ref{Table-dual}. Specifically, $\mathcal{A}_{5}(\mathcal{C})$ is a ternary $3$-$(11,5,1)$ design and $\mathcal{B}_{5}(\mathcal{C})$ is a $4$-$(11,5,1)$ design.
	\end{Example}
	
	\begin{Remark}
		By Theorem \ref{perfect-q(e+1)}, we have $T_{d}\ge T_{d}^{(q)} = e+1$ for linear perfect codes with $d=2e+1$.
		By Example \ref{Ex-golay}, for the ternary Golay code, $\mathcal{A}_{5}(\mathcal{C})$ and $\mathcal{B}_{5}(\mathcal{C})$ are ternary $3$-designs and ordinary $4$-designs respectively, both with index $1$. We say that the minimum weight codewords hold double Steiner systems. This is precious in that it is the only $q$-ary code with $q>2$ having this property to the best of our knowledge.
	\end{Remark}

	\subsection{One- and two-weight codes}
	Linear codes with few distinct weights can be used to construct $q$-ary $t$-designs by the Standard Criterion. We begin with a complete determination of one-weight linear codes that hold $q$-ary $2$-designs.
	The dual of the Hamming codes are called \emph{simplex codes}.
	It is well-known that the simplex code has parameters $[(q^{m}-1)/(q-1),m,q^{m-1}]_{q}$.
	A linear code with only one nonzero weight is called an equidistant or constant-weight code. The structure of such codes is completely known.
	%一个重量的码都是由simplex码产出
	\begin{Lemma}\label{one-weight}
		\cite{bonisoli}
		Any equidistant linear code over $\mathbb{F}_{q}$ can be obtained, up to monomial equivalence, by replicating a $q$-ary simplex code and possibly appending zero coordinates.
	\end{Lemma}
	\begin{Theorem}  \label{Simplex-q}
		Let $\mathcal{C}$ be the $q$-ary simplex code with parameters $\left[n, m, d\right]_q$ where $n = (q^m - 1)/(q-1)$ and $d = q^{m-1}$. Then  $\mathcal{A}_{d}(\mathcal{C})$ forms a $q$-ary $2$-$(n, q^{m-1},q^{m-2})$ design and $\mathcal{B}_{d}(\mathcal{C})$ is a $2$-$(n,q^{m-1},(q-1)q^{m-2})$ design.
	\end{Theorem}
	\begin{proof}
		The dual code $\mathcal{C}^\perp$ is a Hamming code with parameters $[n,n-m,3]_{q}$, implying $d^{\perp}(\mathcal{C})=3$. Moreover, $\mathcal{C}$ is a constant-weight code, which gives $s(\mathcal{C})=1$.
		By the Standard Criterion, $\mathcal{A}_{d}(\mathcal{C})$ forms a $q$-ary $2$-$(n,d,\lambda)$ design.
		The number of codewords of weight $d$ is $|\mathcal{A}_{d}(\mathcal{C})| = q^{m}-1$. Using the relation $|\mathcal{A}_{d}(\mathcal{C})| = {\lambda \binom{n}{2}(q-1)^{2}}/{\binom{d}{2}}$, we obtain $\lambda = q^{m-2}$. Then applying Proposition \ref{repeat}, we conclude that $\mathcal{B}_{d}(\mathcal{C})$ is a $2$-$(n,q^{m-1},(q-1)q^{m-2})$ design.
	\end{proof}
	Theorem \ref{Simplex-q} rediscovers that $\mathcal{B}_{q^{m-1}}(\mathcal{C})$ for the Simplex code $\mathcal{C}$ is a $2$-design, which has been established by \cite[Theorem 10.23]{Ding}.
	Replication or adding zero coordinates to the simplex code generally destroys this design property. Consequently, we have a complete classification.
	\begin{Proposition}
		For $q\ge 3$, under monomial equivalence, the $q$-ary simplex codes are the only equidistant linear codes that hold $q$-ary $2$-designs.
	\end{Proposition}
	\begin{proof}
		Let $\mathcal{C}$ be an equidistant linear code whose minimum weight codewords form a $q$-ary $2$-design. By Lemma \ref{one-weight}, $\mathcal{C}$ is monomially equivalent to a code constructed from a simplex code by replicating coordinates and possibly appending zero coordinates.
		If zero coordinate is appended, choose a vector $\bm{x}$ of weight $2$ with one nonzero entry in that coordinate. Since every codeword has a zero there, $\bm{x}$ is not covered, violating the $q$-ary $2$-design condition.
		If any coordinate is replicated, let $i,j$ be two positions that are copies of the same original coordinate. For any codeword $\bm{c}$ we have $c_{i}=c_{j}$. Now pick a vector $\bm{x}$ of weight $2$ with support $\{i,j\}$ and with different nonzero symbols in these two positions. Then no codeword can cover $\bm{x}$, again contradicting the $q$-ary $2$-design property.
		Thus, $\mathcal{C}$ must be monomially equivalent to the simplex code itself.
	\end{proof}
	We now turn to two-weight codes. Let $\mathcal{C}$ be an $[n,k,d]_{q}$ linear code with two nonzero weights and dual distance $d^{\perp} \ge 4$.  Then, by the Standard Criterion, we obtain the following:
	\begin{itemize}
		\item The codewords of any fixed nonzero weight in $\mathcal{C}$ form a $q$-ary $(d^{\perp}-2)$-design;
		\item The codewords of any fixed nonzero weight in the dual code $\mathcal{C}^{\bot}$ also form a $q$-ary $(d^{\perp}-2)$-design.
	\end{itemize}
	
%	We will need the following MacWilliams identity.
%	\begin{Lemma}\cite[Theorem 2.4]{Ding}\label{MacWilliams}
%		Let $\mathcal{C}$ be an $[n,k,d]_{q}$ linear code with weight enumerator $A(z) = \sum_{i=0}^{n}A_{i}z^{i}$ and let $A^{\perp}(z)$ be the weight enumerator of its dual code $\mathcal{C^{\perp}}$. Then
%		\begin{align}
%			A^{\perp}(z) = q^{-k}\left (1+(q-1)z\right )^{n}A\left (\frac{1-z}{1+(q-1)z}\right ).
%		\end{align}
%	\end{Lemma}
%	We now employ the MacWilliams identity together with the Standard Criterion to construct $q$-ary $2$-designs.
	
	\begin{Theorem}\label{TF1}
		For an even prime power $q$, let $\cal C$ be a $[q+2,3,q]_{q}$ linear code with weight enumerator
		\begin{align*}
			1+\frac{(q+2)(q^{2}-1)}{2}z^{q}+\frac{q(q-1)^{2}}{2}z^{q+2}.
		\end{align*}
		Then the following hold:
		\begin{itemize}
			\item [(1)] $\mathcal{A}_{q}(\mathcal{C})$ forms a $q$-ary $2$-$(q+2,q,\frac{q}{2})$ design;
			\item [(2)] $\mathcal{A}_{q+2}(\mathcal{C})$ forms a $q$-ary $2$-$(q+2,q+2,\frac{q}{2})$ design;
			\item [(3)] $\mathcal{A}_{4}(\mathcal{C}^{\perp})$ forms a $q$-ary $2$-$(q+2,4,\frac{q}{2})$ design.
		\end{itemize}
	\end{Theorem}
	\begin{proof}
		Note that $\cal C$ is an MDS code,
		the dual code $C^{\perp}$ is also MDS and has parameters $[q+2,q-1,4]_{q}$. So, the number of codewords of weight $4$ in $C^{\perp}$ is $\binom{q+2}{4}(q-1)$.
		Thus $d(\mathcal{C^{\perp}}) = 4$. Since  $s(\mathcal{C}) = 2$, by the Standard Criterion, $\mathcal{A}_{q}(\mathcal{C})$, $\mathcal{A}_{q+2}(\mathcal{C})$, and $\mathcal{A}_{4}(\mathcal{C}^{\perp})$ each form a $q$-ary $2$-design with some index $\lambda_{1},\lambda_{2},\lambda_{3}$, respectively.
		For a $q$-ary $2$-$(q+2,w,\lambda)$ design, the relation $|\mathcal{A}_w| = \lambda \binom{q+2}{2}(q-1)^2 / \binom{w}{2}$ holds. Solving for $\lambda$ with the known $w$ and $|\mathcal{A}_w|$ yields $\lambda_{1} = \lambda_{2} = \lambda_{3} = q/2$, completing the proof.
	\end{proof}
	\begin{Theorem} \label{TF3}
		For $q\ge 4$, let $\cal C$ a $[q^{2}+1,4,q^{2}-q]_{q}$ linear code with weight enumerator
		\begin{align*}
			1+(q^{2}-q)(q^{2}+1)z^{q^{2}-q}+(q-1)(q^{2}+1)z^{q^{2}}.
		\end{align*}
		Then the following hold:
		\begin{itemize}
			\item [(1)] $\mathcal{A}_{q^{2}-q}(\mathcal{C})$ forms a $q$-ary $2$-$(q^{2}+1,q^{2}-q,q^{2}-q-1)$ design and $\mathcal{B}_{q^{2}-q}(\mathcal{C})$ forms a $3$-$(q^{2}+1,q^{2}-q,q^{3}-3q^{2}+q+2)$ design;
			\item [(2)] $\mathcal{A}_{q^{2}}(\mathcal{C})$ forms a $q$-ary $2$-$(q^{2}+1,q^{2},q+1)$ design and $\mathcal{B}_{q^{2}}(\mathcal{C})$ is a complete design;
			\item [(3)] $\mathcal{A}_{4}(\mathcal{C}^{\perp})$ forms a $q$-ary $2$-$(q^{2}+1,4,\frac{(q+1)(q-2)}{2})$ design and	$\mathcal{B}_{4}(\mathcal{C}^{\perp})$ forms a $3$-$(q^{2}+1,4,q-2)$ design.
		\end{itemize}
	\end{Theorem}
	\begin{proof}
		From \cite[Chapter 13]{Ding}, the dual code $\mathcal{C}^{\perp}$ must have parameters $[q^{2}+1,q^{2}-3,4]_{q}$ and its number of codewords of weight $4$ is $(q^{2}+1)q^{2}(q-1)^{2}(q+1)(q-2)/24$.
		Hence, $\cal C$ and $\mathcal{C^{\perp}}$ hold $q$-ary $2$-design at every nonzero weight by the Standard Criterion. The indexes are easily computed, as stated.
		Applying the Assmus-Mattson theorem, \cite[Theorem 13.12]{Ding} shows that $\mathcal{B}_{q^{2}-q}(\mathcal{C})$ forms a $3$-$(q^{2}+1,q^{2}-q,q^{3}-3q^{2}+q+2)$ design and $\mathcal{B}_{4}(\mathcal{C}^{\perp})$ forms a $3$-$(q^{2}+1,4,q-2)$ design.
		It remains to consider $\mathcal{B}_{q^{2}}(\mathcal{C})$.
		By Lemma $\ref{repeat_q-1}$, each support of a codeword of weight $q^{2}$ in $\mathcal{C}$ appears with multiplicity $q-1$.
		Since $|A_{q^{2}}(\mathcal{C})| = (q-1)(q^{2}+1)$. It follows that $\mathcal{B}_{q^{2}}(\mathcal{C})$ consists of all $q^{2}$-subsets; hence $\mathcal{B}_{q^{2}}(\mathcal{C})$ is a complete design.
	\end{proof}
	
	\begin{Remark}
		The code in Theorem \ref{TF1} can be constructed from a hyperoval in $\mathrm{PG}(2,q)$, which is denoted by $\mathrm{TF1}$ in \cite{two}. The code in Theorem \ref{TF3} can be constructed from an ovoid in $\mathrm{PG}(3,q)$, denoted by $\mathrm{TF3}$ in \cite{two}.
	\end{Remark}

	The known $q$-ary two-weight linear codes with dual distance $d^\perp \ge 4$ have been completely classified; see
	\cite[Theorem 7.1.1]{SRGraphs}.
	For $q > 2$ and $t\ge 2$, Table \ref{Table-2weight} summarizes all known instances where these codes support $q$-ary $t$-designs at their two nonzero weights $w_1$ and $w_2$. The table lists the parameters for both the classical designs $\mathcal{B}_{w}(\mathcal{C})$ and the $q$-ary designs $\mathcal{A}_{w}(\mathcal{C})$. The codes are named according to \cite[Figure 1]{two}. Correspondingly, Table \ref{Table-dual} presents the parameters of the $q$-ary $t$-designs held by the minimum weight codewords of the dual codes $\mathcal{C}^\perp$.

	\begin{table}[ht]
		\caption{Two-weight linear codes holding $q$-ary $t$-designs ($q>2$, $d^{\perp}\ge 4$)}\vspace{0.2cm}
		\scriptsize
		\centering
		\renewcommand\arraystretch{1.1}
			\begin{tabular}{cccccc}
				\hline
				Code & $[n,k,d]_{q}$ & $\mathcal{A}_{w_{1}}(\mathcal{C})$ & $\mathcal{B}_{w_{1}}(\mathcal{C})$ & $\mathcal{A}_{w_{2}}(\mathcal{C})$ & $\mathcal{B}_{w_{2}}(\mathcal{C})$ \\
				\hline
				\thead{RT6 \\ (Golay)} & $[11,5,6]_{3}$ & $3$-$(11,6,2)_{3}$ & $4$-$(11,6,3)$ & $3$-$(11,9,7)_{3}$ & \thead{ complete \\design} \\
				\thead{FE2 \\ (Hill \cite{FE2}) }& $[56,6,36]_{3}$ & $2$-$(56,36,63)_{3}$ & $2$-$(56,36,126)$ & $2$-$(56,45,18)_{3}$ &  $2$-$(56,45,36)$ \\
				\thead{FE3 \\ (Hill \cite{FE3})} & $[78,6,56]_{4}$ & $2$-$(78,56,160)_{4}$ & $2$-$(78,56,480)$ & $2$-$(78,64,96)_{4}$ &  $2$-$(78,64,288)$ \\
				\thead{TF1 \\(hyperoval) }& \thead{$[q+2,3,q]_{q}$ \\ $q$ even} & $2$-$(q+2,q,\frac{q}{2})_{q}$ &\thead{ complete \\design} & $2$-$(q+2,q+2,\frac{q}{2})_{q}$  & \thead{ complete \\design} \\
				\thead{TF3 \\(ovoid)} & \thead{$[q^{2}+1,4,d]_q$ \\$d=q^{2}-q$\\ $q\ge4$} & \thead{$2$-$(q^{2}+1,d,\lambda)_{q}$\\ $\lambda = q^{2}-q-1$} & \thead{$3$-$(q^{2}+1,d,\lambda)$\\ $\lambda = q^{3}-3q^{2}+q+2$} &
				\thead{$2$-$(q^{2}+1,q^{2},\lambda)_{q}$ \\ $\lambda = q+1$}  & \thead{ complete \\design}  \\
				\hline
			\end{tabular}
			\label{Table-2weight}
		\end{table}
		
		\begin{table}[ht]
			\caption{Dual codes of two-weight linear codes holding $q$-ary $t$-designs ($q>2$, $d^{\perp}\ge 4$)}\vspace{0.2cm}
			\footnotesize
			%		\small
			\centering
			\renewcommand\arraystretch{1.1}
			\begin{tabular}{cccc}
				\hline
				Code &  $[n,k,d]_{q}$ & $\mathcal{A}_{d}(\mathcal{C}^{\perp})$ & $\mathcal{B}_{d}(\mathcal{C}^{\perp})$ \\
				\hline
				RT6$^{\bot}$  & $[11,6,5]_{3}$ & $3$-$(11,5,1)_{3}$ & $4$-$(11,5,1)$ \\
				FE2$^{\bot}$  & $[56,50,4]_{3}$ & $2$-$(56,4,9)_{3}$ & $2$-$(56,4,18)$ \\
				FE3$^{\bot}$  & $[78,72,4]_{4}$ & $2$-$(78,4,6)_{4}$ & $2$-$(78,4,18)$ \\
				TF1$^{\bot}$  & \thead{$[q+2,q-1,4]_{q}$\\ $q$ even} & $2$-$\bigl(q+2,4,\frac{q}{2}\bigr)_{q}$ & \thead{ complete \\design} \\
				TF3$^{\bot}$  & \thead{$[q^{2}+1,q^{2}-3,4]_{q}$ \\ $q\ge4$ } & $2$-$\bigl(q^{2}+1,4,\frac{(q+1)(q-2)}{2}\bigr)_{q}$ & $3$-$(q^{2}+1,4,q-2)$ \\
				\hline
			\end{tabular}
			\label{Table-dual}
		\end{table}

		\noindent
		\textbf{Open problem:} For $q > 2$, determine whether there exists a $q$-ary two-weight linear code with dual distance $3$ holding a $q$-ary $t$-design with $t\ge 2$.
		
		By applying the Puncturing-Shortening Criterion to the dual of the code $\mathrm{TF1}$, we derive the following classes of $q$-ary $2$-designs.
		\begin{Theorem}
			Let $\mathcal{C}$ be the code $\mathrm{TF1}^{\perp}$ with parameters $[q+2,q-1,4]_{q}$ where $q>2$ is even.
			For any coordinate position $m$,
			consider the punctured code $\mathcal{C}^{\{m\}}$ and the shortened code $\mathcal{C}_{\{m\}}$. Then the set of minimum weight codewords of $\mathcal{C}^{\{m\}}$
			forms a $q$-ary $2$-$(q+1,3,1)$ design, and the set of minimum weight codewords of $\mathcal{C}_{\{m\}}$ forms a $q$-ary $2$-$(q+1,4,(q-2)/2)$ design.
		\end{Theorem}
		\begin{proof}
			It is easy to have that $\mathcal{C}^{\{m\}}$ has parameters $[q+1,q-1,3]$ and $\mathcal{C}_{\{m\}}$ has parameters $[q+1,q-2,4]$.
			By Theorem \ref{TF1}(3),  $\mathcal{A}_{4}(\mathcal{C})$ forms a $q$-ary $2$-$(q+2,4,{q}/{2})$ design.
			The dual code $\mathcal{C^{\perp}}$ is $\mathrm{TF1}$, whose nonzero weights are $q$ and $q+2$. Hence $n=q+2$ belongs to $W(\mathcal{C^{\perp}})$, $d(\mathcal{C})=4$ and  $s^{\perp}(\mathcal{C})=2$. Applying the Puncturing-Shortening Criterion to $\mathcal{C}$, we conclude that $\mathcal{A}_{3}(\mathcal{C}^{\{m\}})$ is a $q$-ary $2$-$(q+1,3,1)$ design and that $\mathcal{A}_{4}(\mathcal{C}_{\{m\}})$ is a $q$-ary $2$-$(q+1,4,(q-2)/2)$ design.
		\end{proof}

	\subsection{Extremal self-dual codes}
		Extremal self-dual codes, which meet certain theoretical bounds, are another source of $q$-ary $t$-designs. We focus on two primary families: Type III codes over $\mathbb{F}_3$ and Type IV codes over $\mathbb{F}_4$.
		According to \cite[Section 11.1]{Ding}, an extremal Type III code is a self-dual code over $\mathbb{F}_3$ with parameters $\left [n, n/2, 3\lfloor n/12\rfloor+3 \right ]$ and divisor $\Delta = 3$, where $4 \mid n$.
		An extremal Type IV code is a Hermitian self-dual code over $\mathbb{F}_4$ with parameters $\left [n, n/2, 2\lfloor n/6\rfloor+2\right ]$ and divisor $\Delta = 2$, where $2 \mid n$.
		
		\begin{Theorem}
			\label{type3}
			Let $\mathcal{C}$ be a $[12m+4\mu, 6m+2\mu, 3m+3]$ extremal Type III code for $\mu\in\{0,1,2\}$. Then the codewords of any fixed weight in $\mathcal{C}$ form ternary $t$-designs with the strength $t$ given by:
			\begin{itemize}
				\item [(1)] $t=3$ if $\mu = 0$ and $m\ge 1$,
				\item [(2)] $t=2$ if $\mu = 1$ and $m\ge 0$,
				\item [(3)] $t=1$ if $\mu = 2$ and $m\ge 0$.
			\end{itemize}
		\end{Theorem}
		\begin{proof}
			The number of nonzero weights $s$ in $\mathcal{C}$ satisfies
			\begin{align*}
				s\le  \left \lceil \frac{(12m+4\mu)-(3m+3)+1}{3} \right \rceil = \left \lceil 3m+\frac{4\mu-2}{3} \right \rceil=
				\begin{cases}
					3m,  & \text{ if } \mu=0, \\
					3m+1,& \text{ if } \mu=1, \\
					3m+2,& \text{ if } \mu=2.
				\end{cases}
			\end{align*}
			By the Standard Criterion, the codewords of any fixed weight in $\mathcal{C}$ form a ternary $t$-design.
			Substituting the above upper bounds for $s$ into the criterion yields the stated strengths:
			$t=3$ when $\mu=0$ and $m\ge 1$, $t=2$ when $\mu = 1$ and $m\ge 0$, and $t=1$ when $\mu = 2$ and $m\ge 0$.
		\end{proof}
		
		Below are some examples of extremal Type III codes and their designs.
		\begin{Example}\label{P12}
			The Pless symmetry code $P_{12}$ is an extremal Type III $[12,6,6]_{3}$ code with weight enumerator
			$$1+264z^{6}+440z^{9}+24z^{12}.$$
			This code yields several classical designs, as noted in \cite[Example 11.26]{Ding}.
			Using SageMath, we further determine that its codewords also form a series of ternary $3$-designs, with strength aligning with Theorem \ref{type3}.
			\begin{align*}
				\text{Classical design parameters: } & 5\text{-}(12,6,1),\ 5\text{-}(12,9,35),\ 5\text{-}(12,12,1); \\
				\text{ternary design parameters: } & 3\text{-}(12,6,3)_{3},\ 3\text{-}(12,9,21)_{3},\ 3\text{-}(12,12,3)_{3}.
			\end{align*}
		\end{Example}
		
		\begin{Example}
			The Pless symmetry code $P_{24}$ is an extremal Type III $[24,12,9]_{3}$ code with weight enumerator
			$$1+4048z^{9}+61824z^{12}+242880z^{15}+198352z^{18}+24288z^{21}+48z^{24}.$$
			This code yields several classical designs, as noted in \cite[Example 11.27]{Ding}.
			Using SageMath, we further determine that its codewords also form a series of ternary $3$-designs, with strength aligning with Theorem \ref{type3}.
			\begin{align*}
				\text{Classical design parameters: } & 5\text{-}(24,9,6),\ 5\text{-}(24,12,576),\ 5\text{-}(24,15,8580),\\
				& 3\text{-}(24,18,29784),\ 5\text{-}(24,21,969),\ 5\text{-}(24,24,1); \\
				\text{ternary design parameters: } & 3\text{-}(24,9,21)_{3},\ 3\text{-}(24,12,840)_{3},\ 3\text{-}(24,15,6825)_{3},\\
				& 3\text{-}(24,18,9996)_{3},\ 3\text{-}(24,21,1995)_{3},\ 3\text{-}(24,24,6)_{3}.
			\end{align*}
		\end{Example}

		\begin{Theorem}\label{type4}
			Let $\mathcal{C}$ be a $[6m, 3m, 2m+2]_{4}$ extremal Type IV code for $m\ge 2$. Then for every nonzero weight $w$, $\mathcal{A}_{w}(\mathcal{C})$ is a quaternary $2$-design.
		\end{Theorem}
		\begin{proof}
			The number of distinct nonzero weights in $\cal C$ is at most
			$\left \lceil \frac{6m-(2m+2)+1}{2} \right \rceil = 2m$.
			Hence, by the Standard Criterion, $\mathcal{A}_{w}(\mathcal{C})$ is a quaternary $2$-design for each nonzero weight $w$.
		\end{proof}
		\begin{Example}
			The extension of quadratic residue code $\overline{QRC}_{0}^{(29,4)}$ is an extremal Type IV $[30,15,12]_4$ code with weight enumerator
			\begin{align*}
				1+118755z^{12}+1151010z^{14}+12038625z^{16}+61752600z^{18}+195945750z^{20}+\\
				341403660z^{22}+312800670z^{24}+129570840z^{26}+18581895z^{28}+378018z^{30}.
			\end{align*}
			This code holds a $5$-$(30,12,220)$ design, as noted in \cite[Example 11.22]{Ding}.
			Using SageMath, we further determine that its codewords of weight $12$ also form a quaternary $2$-$(30,12,2002)$ design, with strength aligning with Theorem \ref{type4}.
		\end{Example}
		The content detailed in this section and the previous one demonstrates the use of general criteria.  However, some interesting families of linear code still go beyond their scope.  In the next section, we will develop new approaches to handle these situations.

	\section{Applications of automorphism groups}
		\setcounter{equation}{0}
		While the general criteria established in section $3$ are powerful, they do not guarantee to provide the largest possible strength $t$ for their $q$-ary designs. This section develops novel techniques to prove the existence of $q$-ary $2$-designs in two families of linear codes: doubly-extended Reed-Solomon codes and a type of trace codes.
		A key strategy for establishing these structures exploits the transitive automorphism groups of the codes.
		While $t$-transitivity ensures the existence of an ordinary $t$-design (Theorem \ref{transitive-t-design}), it does not in general guarantee the presence of a $q$-ary $t$-design (Example \ref{transitive-no-q-design}).
		
		\begin{Theorem}
			\label{transitive-design}
			Let $\mathcal{C}\subseteq\mathbb{F}_{q}^{n}$ be a linear code whose automorphism group $\Aut$$(\mathcal{C})$ is $t$-transitive. If for every vector $\bm{x}\in\mathbb{F}_{q}^{n}$ of weight $t$ whose support is a fixed set of $t$ coordinates, there are exactly $\lambda$ codewords in $\mathcal{A}_{w}(\mathcal{C})$ that cover $\bm{x}$, then $\mathcal{A}_{w}(\mathcal{C})$ forms a $q$-ary $t$-$(n,w,\lambda)$ design.
		\end{Theorem}
		\begin{proof}
			Let $S=\{i_{1},i_{2},\dots,i_{t}\}$ be a fixed set of $t$ coordinates positions. Define $X=\{\bm{x}\in\mathbb{F}_{q}^{n} : \mathrm{supp}(\bm{x}) = S\}$. Then $|X|=(q-1)^{t}$. By hypothesis, for any $\bm{x}\in X$ there are exactly $\lambda$ codewords in $\mathcal{A}_{w}(\mathcal{C})$ that cover $\bm{x}$.
			Now take an arbitrary vector $\bm{y}\in\mathbb{F}_{q}^{n}$ of weight $t$. Because $\Aut$$(\mathcal{C})$ is $t$-transitive, there exists an automorphism $h\in$ $\Aut$$(\mathcal{C})$ and a vector $\bm{x}\in X$ such that $h(\bm{x})=\bm{y}$.
			Let $\bm{c}_{1},\dots,\bm{c}_{\lambda}$, be the codewords in $\mathcal{A}_{w}(\mathcal{C})$ that cover $\bm{x}$. Their images $h(\bm{c}_{1}),\dots,h(\bm{c}_{\lambda})$ belong to $\mathcal{A}_{w}(\mathcal{C})$ and cover $\bm{y}$. Hence the number $\lambda'$ of codewords in $\mathcal{A}_{w}(\mathcal{C})$ covering $\bm{y}$ satisfies  $\lambda\le \lambda'$.
			Conversely, by the $t$-transitivity again, there exist $h'\in$ $\Aut$$(\mathcal{C})$ and $\bm{x}'\in X$ with $h'(\bm{y})=\bm{x}'$. If $\bm{c}_{1}',\dots,\bm{c}_{\lambda'}'$ are the codewords in $\mathcal{A}_{w}(\mathcal{C})$ covering $\bm{y}$, then $h'(\bm{c}_{1}'),\dots,h'(\bm{c}_{\lambda'}')$ are codewords in $\mathcal{A}_{w}(\mathcal{C})$ covering $\bm{x}'$, which implies $\lambda'\le \lambda$.
			Therefore $\lambda = \lambda'$. Since $\bm{y}$ is arbitrary, $\mathcal{A}_{w}(\mathcal{C})$ is a $q$-ary $t$-$(n,w,\lambda)$ design.
		\end{proof}
		
		Theorem \ref{transitive-design} reduces the problem of proving that a code holds a $q$-ary $t$-design to verifying a local condition on a fixed coordinate set, provided the automorphism group is $t$-transitive. In the following we apply this method to two concrete families of codes.
		
		\subsection{Doubly-extended Reed-Solomon codes}
		Let $P_{k}[x] = \{ f(x) \in \mathbb{F}_{q}[x] : \deg(f(x)) \le k-1 \}$ denote the set of polynomials over $\mathbb{F}_{q}$ of degree less than a positive integer $k$. Take a primitive element $\alpha$ of $\mathbb{F}_{q}$. The \emph{Reed-Solomon} (RS) code is defined as:
		$$\mathcal{C} = \{ (f(1),f(\alpha),f(\alpha ^{2}),\ldots,f(\alpha^{q-2}) ): f \in P_{k}[x] \},$$
		which is a $[q-1, k, q-k]$ linear code over $\mathbb{F}_{q}$. Its doubly-extended version, the DRS code, is obtained by appending two coordinates corresponding to the evaluation of $f$ at $0$ and at infinity:
		$$\mathcal{C}=\{ (f(0),f(1),f(\alpha),f(\alpha ^{2}),\ldots,f(\alpha^{q-2}), f(\infty)) : f \in P_{k}[x] \},$$
		where $f(\infty)$ is defined as the coefficient of $x^{k-1}$ (hence $f(\infty)=0$ if and only if $\deg(f) < k-1$). This DRS code is a linear MDS code with parameters $[q+1, k, q-k+2]_q$. By Theorem \ref{MDS-Q1}, it holds $q$-ary $1$-design at every nonzero weight. We will show that $q$-ary designs with higher strength could be held in the DRS codes.
		
		According to \cite[Remark $3$]{Aut-RS}, the permutation automorphism group of a DRS code over $\mathbb{F}_{q}$ is the projective semilinear group $\mathrm{P\Gamma L}(2,q)$, which acts triply transitive on $\mathrm{PG}(1,q)$. This can be stated as follows.
		\begin{Lemma}\label{DRS-3transitive}
			The automorphism group of a $\mathrm{DRS}$ code is triply transitive.
		\end{Lemma}
		
		The following lemma is a classical result on subset sums in cyclic groups and serves as a tool for analyzing codewords in DRS codes that cover a given vector.
		%√
		\begin{Lemma}
			\label{countNKB}\cite[Corollary 1.2]{CountSum}
			For an element $b\in \mathbb{Z}_{n}$, let $M(k,b)$ denote the number of k-subsets of $\mathbb{Z}_{n}$ whose elements sum to $b$.  Then
			$$M(k,b)=\frac{1}{n}\sum_{r\mid\gcd(n,k)}(-1)^{k+\frac{k}{r}} \binom{n/r}{k/r}C_{r}(b), $$
			where $C_{r}(b)=\sum_{d\mid \gcd(r,b)}\mu(r/d)d$ and $\mu$ is the M$\ddot{o}$bius function on the integers.
		\end{Lemma}
		Under the isomorphism between the multiplicative group $\mathbb{F}_{q}^*$ and the additive group $\mathbb{Z}_{q-1}$, Lemma \ref{countNKB} translates directly to a result on subsets of $\mathbb{F}_{q}^*$ with a prescribed product.
		%计数 √
		\begin{Corollary}
			\label{N(mc)} For any prime power $q$, let $n=q-1$. For a nonzero element $c \in \mathbb{F}^{*}_{q}$, let $N(k,c)$ denote the number of k-subsets of $\mathbb{F}^{*}_{q}$ whose elements multiply to c. If $\gcd(k,n)=1$, then $N(k,c) = \frac{1}{n}\binom{n}{k}$ which is independent of $c$.
		\end{Corollary}
		\begin{proof}		
			The isomorphism $\mathbb{F}_{q}^* \cong \mathbb{Z}_{n}$ maps the product of elements in $\mathbb{F}_{q}^*$ to the sum of their discrete logarithms in $\mathbb{Z}_n$. Thus, $N(k, c) = M(k, b)$, where $b$ is the discrete logarithm of $c$. The condition $\gcd(k, n)=1$ ensures that $M(k, b) = \frac{1}{n}\binom{n}{k}$ by Lemma \ref{countNKB}.
		\end{proof}
		
		Equipped with these combinatorial tools, we now present the main theorem of this subsection.
		%GDRS --> t-design  √
		\begin{Theorem}\label{DRS-q2}
			Let $\mathcal{C}$ be a \textit{DRS} code with length $n=q+1$, dimension $k$, and minimum distance $d=q-k+2$. If $\gcd(k-1,q-1)=1$, then the codewords of weight $d$ in $\mathcal{C}$ form a $q$-ary $2$-$(q+1,q-k+2,\frac{1}{q-1}\binom{q-1}{k-1} )$ design.
		\end{Theorem}
		\begin{proof}
			Let $\bm{c}=(f(0),f(1),f(\alpha),f(\alpha ^{2}),\dots,f(\alpha^{q-2}), f(\infty)) $ be a codeword of weight $d$ in $\mathcal{C}$. This implies that $f$ has $k-1$ distinct zeros in $\mathbb{F}_{q}\cup \{\infty\}$.
			Define the zero sets by
			\[
			Z(f) = \{z\in \mathbb{F}_{q}\cup \{\infty\} : f(z)=0 \}, \quad Z'(f) = \{z\in \mathbb{F}_{q} : f(z)=0 \}.
			\]
			Since $f \in P_k[x]$ is a polynomial of degree at most $k-1$, the condition $|Z(f)|=k-1$ implies that $\deg(f)$ must be either $k-1$ or $k-2$.
			
			By Lemma \ref{DRS-3transitive}, the automorphism group of $\mathcal{C}$ is triply transitive. Moreover, by Lemma \ref{repeat_q-1}, each support of a codeword of weight $d$ in $\mathcal{C}$ corresponds to exactly $q-1$ distinct codewords.
			Combining these facts with Theorem \ref{transitive-design}, we obtain that for any vector $\bm{v}$ of weight two with $v_{i}=a$ and $v_{j} = b$ where $a,b\in\mathbb{F}_{q}^{*}$, the number $\lambda_{\bm{v}}$ of codewords $\bm{c}$ that cover $\bm{v}$ equals
			\begin{align*}
				\lambda_{\bm{v}} &=|\left \{\bm{c} \in \mathcal{C} : c_{i}=a,c_{j}=b,\mathrm{wt}(\bm{c})=d \right \}| \\
				&=|\left \{\bm{c}\in \mathcal{C} : c_{1}=a,c_{2}=b,\mathrm{wt}(\bm{c})=d \right \}| \\
				&=|\{f \in P_{k}[x] : f(0)/f(1)=z, |Z(f)|=k-1 \}|,
			\end{align*}
			where $z = a/b$.
			Denote by $S_{z}$ the set of polynomials in the last equation above.
			We now enumerate such polynomials by considering the degree of $f$.
			
			\noindent\textbf{Case 1:} $\deg(f)=k-1$. Then $f(\infty)\ne 0$, and $|Z'(f)|=k-1$. Let
			\[
			f(x)=\tau \prod_{r\in Z(f)}(x-r), \quad \tau \in \mathbb{F}_{q}^{*}, \quad Z(f) \subseteq \mathbb{F}_{q}\setminus\{0,1\}.
			\]
			Evaluating at $x=0$ and $x=1$ gives
			\[
			f(0)=\tau \prod_{r\in Z(f)}(-r),\quad f(1)=\tau \prod_{r\in Z(f)}(1-r).
			\]
			Dividing these equations yields
			$$\prod_{r\in Z(f)} \frac{r}{1-r}=(-1)^{k-1}z=z'. $$		
			Define the bijection $\phi : \mathbb{F}_{q} \setminus\{0,1\} \to \mathbb{F}_{q} \setminus\{0,-1\}$ by $\phi(r)=\frac{r}{1-r}$. Then
			$$ \prod_{r\in Z(f)} \phi(r) = z'.$$
			Hence, the number $x$ of polynomials $f\in P_{k}[x]$ satisfying that $f\in S_{z}$ and $\deg(f) =k-1$ is		
			\begin{align*}
				x&=\left |\left \{f \in P_{k}[x] : \deg(f)=k-1, f(0)/f(1)=z, |Z(f)|=k-1 \right \}\right | \\
				&=\left |\left  \{Z(f) \subseteq  \mathbb{F}_{q} \setminus\{0,1\} : |Z(f)|=k-1, \textstyle\prod_{r\in Z(f)} \phi(r)=z'\right \}\right |\\
				&=\left |\left \{\phi(Z(f)) \subseteq  \mathbb{F}_{q} \setminus\{0,-1\} : |\phi(Z(f))|=k-1, \textstyle\prod_{r\in \phi(Z(f))} r=z'\right \}\right |\\
				&=\overline{N}(k-1,z'),
			\end{align*}
			where $\overline{N}(k,c)$ denotes the number of $k$-subsets of $\mathbb{F}^{*}_{q}\setminus \{-1\}$ whose elements multiply to $c$.
			
			\noindent\textbf{Case 2:} $\deg(f)=k-2$. Then $f(\infty) = 0$, and $|Z'(f)|=k-2$. Let
			\[
			f(x)=\tau \prod_{r\in Z'(f)}(x-r), \quad \tau \in \mathbb{F}_{q}^{*}, \quad Z'(f) \subseteq \mathbb{F}_{q}\setminus\{0,1\}.
			\]
			Similar to Case $1$, we have
			$$\prod_{r\in Z'(f)} \frac{r}{1-r}=(-1)^{k-2}z=-z'.$$
			Applying the same bijection $\phi$, we obtain that the number $y$ of polynomials $f\in P_{k}[x]$ satisfying that $f\in S_{z}$ and $\deg(f) =k-2$ is	
			\begin{align*}
				y &=\left |\left \{f \in P_{k}[x] : \deg(f)=k-2, f(0)/f(1)=z, |Z(f)|=k-1 \right \}\right | \\
				&=\left |\left \{Z'(f) \subseteq  \mathbb{F}_{q} \setminus\{0,1\} : |Z'(f)|=k-2, \textstyle\prod_{r\in Z'(f)} \phi(r)=-z'\right \}\right |\\
				&=\left |\left \{\phi(Z'(f)) \subseteq  \mathbb{F}_{q} \setminus\{0,-1\} : |\phi(Z'(f))|=k-2, \textstyle\prod_{r\in \phi(Z'(f))} r=-z'\right \}\right |\\
				&=\overline{N}(k-2,-z').
			\end{align*}
			
			Combining both cases yields that
			\begin{align*}
				\lambda_{\bm{v}} = |S_{z}|
				= x+y
				= \overline{N}(k-1,z')+\overline{N}(k-2,-z')
				=N(k-1,z')=\frac{1}{q-1}\binom{q-1}{k-1},
			\end{align*}
			where the last equality follows from Corollary \ref{N(mc)}. Then we conclude the proof.		
		\end{proof}

		\subsection{Trace codes $\mathcal{C}_{\{1,2,3\}}$}
		
		A class of trace codes has been studied \cite{BCH-4-design,Yan}, in which infinite families of $4$-designs are supported.
		We further investigate these codes in this subsection.
		Two infinite families of $q$-ary $2$-designs will be disclosed in the codewords of the two lowest weights.
		In short, this class of trace codes with minimum distance $d$ satisfies that $T_{w}=4$ and $T_{w}^{q}=2$ for $w = d,d+1$.

		For positive integers $k$ and $l$, the \emph{elementary symmetric polynomial} (ESP) of degree $l$ in $k$ variables $u_{1},u_{2},\ldots,u_{k}$ is defined as
		$$\sigma_{l}(u_{1},\ldots,u_{k}) = \sum_{I\subseteq [k],|I|=l }\prod_{j\in I}u_{j}.$$
		$\sigma_{l}(u_{1},\ldots,u_{k})$ is also denoted by $\sigma_{l}(B)$ where $B = \{u_{1},\ldots,u_{k}\}$.
		
		Throughout this subsection, we assume that $q=2^{m}$ where $m\ge 5$ is an odd integer. Let $U=\{u\in \mathbb{F}_{q^{2}}^{*} : u^{q+1}=1\}$ be the subgroup of $\mathbb{F}^{*}_{q^{2}}$ of order $q+1$. For any integer $k$ with $1\le k \le q+1$, let $\binom{U}{k}$ be the set of all $k$-subsets of $U$. Define two families of block sets by these polynomials:
		\begin{align}
			&\mathcal{B}_{k,l} = \left\{ B \in \binom{U}{k} :\sigma_{l}(B) = 0 \right\}. \label{B63}\\
			&\mathcal{B}_{k,l}^{b} = \left\{ B \in \binom{U}{k} :\sigma_{l}(B-a) = 0 \text{ for some }a\in B \right\}, \label{B53}
		\end{align}
		where $B-a = \{u-a : u\in B\}$.
%		For a block $B=\{u_{1},u_{2},\dots,u_{k}\}\in\binom{U}{k}$, the evaluated polynomial can be expanded as
%		$$\sigma_{l}(B-a)=\sum_{I\subseteq [k],|I|=l }\prod_{j\in I}(u_{j}-a) = \sum_{i=0}^{l}(-a)^{l-i}\binom{k-i}{l-i}\sigma_{i}(B).$$
		For a detailed analysis of $\mathcal{B}_{k,l}$ and $\mathcal{B}_{k,l}^{b}$, we refer the reader to \cite{BCH-4-design,Yan}. The following two theorems are key foundational results, establishing that these block sets form classical $4$-designs.
		
		\begin{Lemma}
			\cite[Theorem 2]{BCH-4-design} \label{B63-4design} For an odd integer $m\ge 5$ and $q=2^{m}$, the incidence structure $(U,\mathcal{B}_{6,3})$ is a $4$-$(q+1,6,(q-8)/2)$ design.
		\end{Lemma}
		\begin{Lemma}
			\cite[Theorem 8]{Yan}\label{B53-4design} For an odd integer $m\ge 5$ and $q=2^{m}$, the incidence structure $(U,\mathcal{B}_{5,3}^{b})$ is a $4$-$(q+1,5,5)$ design.
		\end{Lemma}
		
		We now introduce the linear code of interest.
		Let $\alpha$ be a generator of $\mathbb{F}_{q^{2}}^{*}$ and let $\gamma = \alpha^{q-1}$, which generates the cyclic subgroup $U$. Define the trace code $\mathcal{C}_{\{1,2,3\}}$ over $\mathbb{F}_{q}$ as
		\begin{align}\label{C123}
			\mathcal{C}_{\{1,2,3\}} = \{\bm{c}_{(a,b,c)} : a,b,c \in \mathbb{F}_{q^{2}}\},
		\end{align}
		where the codeword is given by
		$$\bm{c}_{(a,b,c)} = (\mathrm{Tr}_{q^{2}/q}(a\gamma^{i}+b\gamma^{2i}+c\gamma^{3i}))_{i=0}^{q}.$$
		By \cite[Theorem 34]{BCH-4-design}, this code has parameters $[q+1,6,q-5]$.
		
		\begin{Lemma}
			\cite[Lemma 33]{BCH-4-design} \label{Tr60} Let $f(u)=\Tr_{q^{2}/q}(au+bu^{2}+cu^{3})$ where $(a,b,c) \in \mathbb{F}_{q}^{3}\setminus \{\mathbf{0}\}$. Define $Z(f)=\{u\in U : f(u)=0\}$. Then $|Z(f)|=6$ if and only if there exist $B=\{u_{1},\dots,u_{6}\} \in \mathcal{B}_{6,3}$ and $\tau \in \mathbb{F}_q^*$ such that $a={\tau\sigma_{2}(B)}\big/{\sqrt{\sigma_{6}(B)}}$, $b={\tau\sigma_{1}(B)}\big/{\sqrt{\sigma_{6}(B)}}$ and $c={\tau}\big/{\sqrt{\sigma_{6}(B)}}$. Under these conditions, $f(u)$ can be written as
			$f(u)=\frac{c}{u^{3}}\prod_{j=1}^{6}(u+u_{j}).$
		\end{Lemma}
		\begin{Lemma}
			\cite[Lemma 29]{Yan} \label{Tr53} Let $f(u)=\Tr_{q^{2}/q}(au+bu^{2}+cu^{3})$ where $(a,b,c) \in \mathbb{F}_{q}^{3}\setminus \{\mathbf{0}\}$. Define $Z(f)=\{u\in U : f(u)=0\}$. Then $|Z(f)|=5$ if and only if there exist $B=\{u_{1},\dots,u_{5}\} \in \mathcal{B}_{5,3}^{b}$, a unique element $u_{i}\in B$ and $\tau \in \mathbb{F}_{q}^{*}$ such that
			$a={\tau(\sigma_{2}(B)+u_{i}\sigma_{1}(B))}\big/{\sqrt{u_{i}\sigma_{5}(B)}}$, $b=c{\tau(\sigma_{1}(B)+u_{i})}\big/{\sqrt{u_{i}\sigma_{5}(B)}}$ and $c={\tau}\big/{\sqrt{u_{i}\sigma_{5}(B)}}$. In this case, $f(u)$ can be written as
			$f(u)=\frac{c(u+u_{i})}{u^{3}}\prod_{j=1}^{5}(u+u_{j}). $
		\end{Lemma}

		Consider the projective general linear group $\PGL(2,q)$ and its natural action on the points of the projective line PG$(1,q)$. An element represented by a matrix
		$\begin{bmatrix}
			a&b \\
			c&d
		\end{bmatrix}\in\PGL(2,q)$ acts on PG$(1,q)$ via
		\[
		(x_{0}:x_{1})\mapsto (ax_{0}+bx_{1} : cx_{0}+dx_{1}).
		\]
		Using the usual identification of $\mathbb{F}_{q}\cup\{\infty\}$ with PG$(1,q)$, this action can also be expressed as the linear fractional transformation
		$x\mapsto\frac{ax+b}{cx+d}$
		where $-d/c$ is mapped to $\infty$, and $\infty$ is mapped to $a/c$ if $c\ne 0$, or to $\infty$ if $c=0$.

		Let $\Stab_{U}=\{g\in\PGL(2,q^{2}) : g(U)=U\}$ denote the setwise stabilizer of the cyclic subgroup $U$ under the action of $\PGL(2,q^{2})$ on PG$(1,q^{2})$.
		The structure of this stabilizer is described in the following lemma.	
		\begin{Lemma}\label{Stab-U-3t}
			\cite[Proposition 6]{3-tran-BCH} The setwise stabilizer $\Stab_{U}$ can be expressed as
			\[
			\left \{
			\begin{pmatrix}
				\beta_{2}^{q}&\beta_{1}^{q} \\
				\beta_{1}&\beta_{2}
			\end{pmatrix}\in \PGL(2,q^{2}) : \beta_{1}^{q+1}\ne \beta_{2}^{q+1}	
			\right \}.
			\]
			Moreover, $\mathrm{Stab}_{U}$ is conjugate in $\mathrm{PGL}(2,q^2)$ to $\mathrm{PGL}(2,q)$, and its action on $U$ is equivalent to the triply transitive action of $\mathrm{PGL}(2,q)$ on $\mathrm{PG}(1,q)$.
		\end{Lemma}
		
		Using Lemmas 8-10 in \cite{3-tran-BCH} and following the same argument as that in the proof of Theorem 21 therein, we have the following lemma.
		
		\begin{Lemma}\label{3transitive}
			Define a subgroup $G$ of $(\mathbb{F}_{q}^{*})^{q+1} \rtimes$ $\Stab_{U}$ by
			\[
			\left \{\left (
			\left ((\beta_{1}u+\beta_{2})^{3(q+1)}\right )_{u\in U} ; \begin{pmatrix}
				\beta_{2}^{q}&\beta_{1}^{q} \\
				\beta_{1}&\beta_{2}
			\end{pmatrix}^{-1}
			\right ) : \beta_{1},\beta_{2}\in\mathbb{F}_{q^{2}},\beta_{1}^{q+1}\ne \beta_{2}^{q+1}\right \}.
			\]
			Then $G$ is a subgroup of $\MAut(\mathcal{C}_{\{1,2,3\}})$. In particular, the automorphism group of $\mathcal{C}_{\{1,2,3\}}$ is triply transitive.
		\end{Lemma}

		Equipped with these results on the code $\mathcal{C}_{\{1,2,3\}}$, we now present two main theorems of this subsection.
		\begin{Theorem}
			\label{C123design}
			Let $m\ge 5$ be an odd integer, $q=2^{m}$, and let $\mathcal{C}_{\{1,2,3\}}$ be defined by \eqref{C123}. Then the set of codewords with minimum weight in $\mathcal{C}_{\{1,2,3\}}$ forms a $q$-ary $2$-$(q+1,q-5,\lambda)$ design, where $\lambda = (q-2)(q-5)(q-6)(q-8)/{6!} $.
		\end{Theorem}
		\begin{proof}
			Let $\bm{c}_{(a,b,c)}=(\Tr_{q^{2}/q}(a\gamma^{i}+b\gamma^{2i}+c\gamma^{3i}))_{i=0}^{q} \in \mathcal{C}_{\{1,2,3\}}$ be a codeword of minimum weight $d=q-5$. Then exactly six positions in $\bm{c}_{(a,b,c)}$ are zeros.
			Correspondingly, identifying $\gamma^{i}$ with $x$, we have that $f(x) = \Tr_{q^{2}/q}(ax+bx^{2}+cx^{3})$ has six distinct zeros in $U$. By Lemma \ref{Tr60}, we express $f(x)$ in terms of its zeros:
			$$f(x)=\frac{c}{x^{3}}\prod_{j=1}^{6}(x+u_{j}),\quad x\in U,$$
			where $\{u_{1},u_{2},\ldots,u_{6}\}\in \mathcal{B}_{6,3}$. The codeword can be written as $$\bm{c}=(f(1),f(\gamma),f(\gamma^{2}),\dots,f(\gamma^{q})).$$
			Thus, each 6-element set $\{u_{1},\dots,u_{6}\}\in \mathcal{B}_{6,3}$ corresponds to a function $f(x)$ and a codeword of weight $d$.
			
			The automorphism group of $ \mathcal{C}_{\{1,2,3\}}$ is triply transitive by Lemma \ref{3transitive}.
			Furthermore, each support of a minimum weight codeword is shared by exactly $q-1$ distinct codewords.
			Combining these facts with Theorem \ref{transitive-design}, for any weight two vector $\bm{v}$ with $v_{i}=a$ and $v_{j}=b$ where $a,b\in\mathbb{F}_{q}^{*}$, $1\le i<j\le q+1$, we obtain that the number $\lambda_{\bm{v}}$ of codewords $\bm{c}$ that cover $\bm{v}$ equals
			\begin{align*}
				\lambda_{\bm{v}} &=\left |\left \{\bm{c}\in \mathcal{C}_{\{1,2,3\}} : c_{i}=a,c_{j}=b,\mathrm{wt}(\bm{c})=d\right \}\right |\\
				&=\left |\left \{\bm{c}\in \mathcal{C}_{\{1,2,3\}} : c_{2}/c_{q+1} = f(\gamma)/f(\gamma^{q}) =z,\mathrm{wt}(\bm{c})=d\right \}\right |\\
				%				&=\left |\left \{\bm{c}\in \mathcal{C}_{\{1,2,3\}} : f(\gamma)/f(\gamma^{q})=z,\mathrm{wt}(\bm{c})=d \right \}\right | \\
				&=\left |\left\{B\in \binom{U}{6} : \sigma_{3}(B)=0, B\cap \{\gamma,\gamma^{-1}\}=\emptyset, f(\gamma)/f(\gamma^{-1})=z\right\}\right |,
			\end{align*}
			where $z=a/b$.
			Evaluating $f(x)$ at $x=\gamma$ and $x=\gamma^{-1}$ yields that
			\[
			f(\gamma)=\frac{c}{\gamma^{3}}\prod_{j=1}^{6}(\gamma+u_{j}),\quad f(\gamma^{-1})=\frac{c}{\gamma^{3}}\prod_{j=1}^{6}(1+\gamma u_{j}).
			\]
			Thus
			$$\dfrac{f(\gamma)}{f(\gamma^{-1})}=\prod_{u\in B} \frac{\gamma+u}{1+\gamma u}=z, $$	
			where $B=\{u_{1},\dots,u_{6}\}\in \mathcal{B}_{6,3} $ and $B\cap \{\gamma,\gamma^{-1}\}=\emptyset$.

			Consider the linear transformation in $\mathrm{GL}(2,q^{2})$ defined by $x\mapsto\dfrac{\gamma+x}{1+\gamma x}$. Denote by $\phi$ its restriction to $U\backslash \{\gamma,\gamma^{-1}\}$. It is easy to see that $\phi$ is a bijection from $U\backslash \{\gamma,\gamma^{-1}\}$ to $\mathbb{F}_{q}^{*}$.
			
			Define
			$$\Omega = \left\{ B\in \binom{U}{6} : \sigma_{3}(B)=0, B\cap \{\gamma,\gamma^{-1}\}=\emptyset \right\}. $$	
			By Lemma \ref{B63-4design}, $\mathcal{B}_{6,3}$ is a $4$-$(q+1,6,(q-8)/2)$ design.
			Then by \eqref{lambda-i} and the principle of inclusion-exclusion, we have
			\begin{align*}
				|\Omega|=\frac{q-8}{2}\left (\frac{\binom{q+1}{4}}{\binom{6}{4}}-2\frac{\binom{q}{3}}{\binom{5}{3}}+\frac{\binom{q-1}{2}}{\binom{4}{2}}\right )
				= (q-1)(q-2)(q-5)(q-6)(q-8)/6!.
			\end{align*}
			Define
			$$\Gamma = \left\{Y\in \binom{F_{q}^{*}}{6} : \phi^{-1}(Y)\in \Omega \right\}.$$
			For $Y=\{y_{1},\dots,y_{6}\}\in \Gamma$, let $u_{i}=\phi^{-1}(y_{i})$. Then $\sigma_{3}(u_{1},\dots,u_{6})=0$.
			The following claim is crucial.
			
			\noindent\textbf{Claim: }If $Y\in\Gamma$, then $hY\in\Gamma$ for any $h\in\mathbb{F}_{q}^{*}$.
			
			\noindent\textit{Proof of the claim. }
			For $h\in \mathbb{F}_{q}^{*}$ and $Y\in\Gamma$, let $v_{i} = \phi^{-1}(hy_{i})$. Substituting $ y_{i}=\phi(u_{i})=\frac{\gamma+u_{i}}{1+\gamma u_{i}}$ into the expression for $v_{i}$ yields:
			\[
			v_{i} = \frac{\beta_{2}^{q}u_{i}+\beta_{1}^{q}}{\beta_{1}u_{i}+\beta_{2}}, \quad \text{where } \beta_{1} = h + 1,\ \beta_{2} = \gamma^{-1}+\gamma h.
			\]
			
			Define $\psi$ on $U\setminus\{\gamma,\gamma^{-1}\}$ by $\psi(x) = \frac{\beta_{2}^{q}x+\beta_{1}^{q}}{\beta_{1}x+\beta_{2}}$.
			Because $\beta_{1}^{q+1}\ne\beta_{2}^{q+1}$,
			$\psi$ is a bijection from $U\setminus\{\gamma,\gamma^{-1}\}$ to itself, fixing $\gamma$ and $\gamma^{-1}$.
			Then $\psi\in\Stab_{U}$ by Lemma \ref{Stab-U-3t}.
			By Lemma \ref{3transitive}, $\psi$ is the permutation part of a monomial automorphism $g\in\MAut(\mathcal{C}_{\{1,2,3\}})$
			that preserves codewords of weight $d$. Consequently, if a codeword $\bm{c}_{1}\in\mathcal{C}_{\{1,2,3\}}$ of weight $d$ has its zero set represented by $\{u_1, \dots, u_6\} \in \Omega$, then the image codeword $\bm{c}_{2} = g(\bm{c}_{1})$ is also of weight $d$ and corresponds to the zero set $\{v_1, \dots, v_6\} =  \{\psi(u_1), \dots, \psi(u_6)\}\in\mathcal{B}_{6,3}$. 	
			Moreover $\{v_1, \dots, v_6\} \in \Omega$.
			Since $hY = \{h y_1, \dots, h y_6\} = \{\phi(v_1), \dots, \phi(v_6)\}$, we conclude $hY \in \Gamma$.

			The map $\mu : \mathbb{F}_{q}^{*}\to \mathbb{F}_{q}^{*} $ defined by $\mu(x)=x^{6}$ is a bijection since $\gcd(6,q-1)=1$.
			Given $z_{1}\in\mathbb{F}_{q}^{*}$ and $Y_{1}\in \Gamma$ with $\prod_{y\in Y_{1}}y=z_{1}$, for any $z_{2}\in \mathbb{F}_{q}^{*}$, choose $h\in \mathbb{F}_{q}^{*}$ such that $h^{6}=z_{2}z_{1}^{-1}$. Then let $Y_{2}=hY_{1}$. By the claim, $Y_{2}\in \Gamma$, so we have $$\prod_{y\in Y_{2}}y=\prod_{y\in Y_{1}}hy=h^{6}\prod_{y\in Y_{1}}y=z_{2}.$$
			This establishes a bijection between the two sets $\Gamma_{z_{1}}$ and $\Gamma_{z_{2}}$, where
			$\Gamma_{x} = \left \{ Y\in\Gamma :  \prod_{y\in Y}y=x \right \}$ for $x\in\mathbb{F}_{q}^{*}$.
			Since $|\Gamma|=|\Omega|$, we obtain
			\begin{align*}
				\lambda_{\bm{v}} &=\left |\left \{ B\in\Omega :  \textstyle\prod_{u\in B}\frac{\gamma+u}{1+\gamma u}= z\right \}\right | \\
				&= \left |\left \{ Y\in\Gamma :  \textstyle\prod_{y\in Y}y=z\right \}\right | \\
				&= |\Omega|/{(q-1)} \\
				&= {(q-2)(q-5)(q-6)(q-8)}/{6!}.
			\end{align*}
			Thus the required parameters hold, completing the proof.
		\end{proof}
		
		Next, we show that the set of codewords of weight $d+1$ in $\mathcal{C}_{\{1,2,3\}}$ also forms a $q$-ary $2$-design.
		\begin{Theorem}
			\label{C123-d+1}
			Let $m\ge 5$ be an odd integer, $q=2^{m}$, and let $\mathcal{C}_{\{1,2,3\}}$ be defined by \eqref{C123}. Then the set of codewords of weight $q-4$ in $\mathcal{C}_{\{1,2,3\}}$ forms a $q$-ary $2$-$(q+1,q-4,\lambda)$ design, where $\lambda = (q-2)(q-4)(q-5)/{4!}$.
		\end{Theorem}	
		\begin{proof}
			We adopt a methodology analogous to the proof of Theorem \ref{C123design}.
			Let $\bm{c}_{(a,b,c)}=(\Tr_{q^{2}/q}(a\gamma^{i}+b\gamma^{2i}+c\gamma^{3i}))_{i=0}^{q} \in \mathcal{C}_{\{1,2,3\}}$ be a codeword of weight $d+1 = q-4$. Then exactly five positions in $\bm{c}_{(a,b,c)}$ are zeros.
			Correspondingly, identifying $\gamma^{i}$ with $x$, we have that $f(x) = \Tr_{q^{2}/q}(ax+bx^{2}+cx^{3})$ has five distinct zeros in $U$.			
			By Lemma \ref{Tr53}, we write $f(x)$ in terms of its zeros:
			$$f(x)=\frac{c(x+u_{i})}{x^{3}}\prod_{j=1}^{5}(x+u_{j}),\quad x\in U,$$
			where $\{u_{1},\dots,u_{5}\}\in \mathcal{B}_{5,3}^{b}$ and $u_{i}\in \{u_{1},\dots,u_{5}\}$ is fixed. The codeword can be written as $$\bm{c}=(f(1),f(\gamma),f(\gamma^{2}),\dots,f(\gamma^{q})).$$

			By Lemma \ref{repeat_q-1}, Theorem \ref{transitive-design} and Lemma \ref{3transitive}, it follows that for any weight two vector $\bm{v}$ with $v_{i}=a$ and $v_{j}=b$ where $a,b\in\mathbb{F}_{q}^{*}$, we obtain that the number $\lambda_{\bm{v}}$ of codewords $\bm{c}$ that cover $\bm{v}$ equals
			\begin{align*}
				\lambda_{\bm{v}} &=\left |\left \{\bm{c}\in \mathcal{C}_{\{1,2,3\}} : c_{i}=a,c_{j}=b,\mathrm{wt}(\bm{c})=d+1\right \}\right |\\
				&=\left |\left \{\bm{c}\in \mathcal{C}_{\{1,2,3\}} : c_{2}/c_{q+1} = f(\gamma)/f(\gamma^{q}) =z,\mathrm{wt}(\bm{c})=d+1 \right \}\right | \\
				&=\left |\left\{ B\in\mathcal{B}_{5,3}^{b} : B\cap \{\gamma,\gamma^{-1}\}=\emptyset, f(\gamma)/f(\gamma^{-1})=z\right\}\right |
			\end{align*}
			where $z = a/b$.
			Evaluating $f(x)$ at $x=\gamma$ and $x=\gamma^{-1}$ yields
			\[
			f(\gamma)=\frac{c(\gamma + u_{i})}{\gamma^{3}}\prod_{j=1}^{5}(\gamma+u_{j}),\quad f(\gamma^{-1})=\frac{c(1+\gamma u_{i})}{\gamma^{3}}\prod_{j=1}^{5}(1+\gamma u_{j}),\quad
			\frac{\gamma+u_{i}}{1+\gamma u_{i}}\prod_{u\in B} \frac{\gamma+u}{1+\gamma u}=z,
			\]
			where $B=\{u_{1},\dots,u_{5}\}\in \mathcal{B}_{5,3}^{b}$ and $B\cap \{\gamma,\gamma^{-1}\}=\emptyset$.
			Define the bijection $\phi : U\backslash \{\gamma,\gamma^{-1}\} \to \mathbb{F}_{q}^{*} $ by
			$\phi(u)=\frac{\gamma+u}{1+\gamma u}.$ Let
			$$\Omega = \left\{ B\in \mathcal{B}_{5,3}^{b} : B\cap \{\gamma,\gamma^{-1}\}=\emptyset \right\},\quad \Gamma = \left\{ Y\in \binom{F_{q}^{*}}{5} : \phi^{-1}(Y)\in \Omega \right\}. $$
			By Lemma \ref{B53-4design}, $\mathcal{B}_{5,3}^{b}$ is a $4$-$(q+1,5,5)$ design. Then by \eqref{lambda-i} and  the principle of inclusion-exclusion, we have
			$$|\Omega|=(q-1)(q-2)(q-4)(q-5)/4!.$$
			
			For arbitrary $h\in \mathbb{F}_{q}^{*}$ and $Y=\{y_{1},\dots,y_{5}\}\in\Gamma$, we have $hY \in \Gamma$.	
			Furthermore, the map $\mu : \mathbb{F}_{q}^{*}\to \mathbb{F}_{q}^{*} $ defined by $\mu(x)=x^{6}$ is a bijection.
			For $Y_{1}\in \Gamma$ with $y_{i}\prod_{y\in Y_{1}}y=z_{1}\in\mathbb{F}_{q}^{*}$, $y_{i}\in Y_{1}$, and for any $z_{2}\in \mathbb{F}_{q}^{*}$, choose $h\in \mathbb{F}_{q}^{*}$ such that $h^{6}=z_{2}z_{1}^{-1}$. Then $Y_{2}=hY_{1}\in \Gamma$, and
			$$y_{i}'\prod_{y\in Y_{2}}y=hy_{i}\prod_{y\in Y_{1}}hy=h^{6}y_{i}\prod_{y\in Y_{1}}y=z_{2}.$$
			Since $|\Gamma|=|\Omega|$, we conclude that
			\begin{align*}
				\lambda_{v} &= \left |\left\{ B\in\Omega :  f(\gamma)/f(\gamma^{-1})=z\right\}\right | \\
				&=\left |\left \{ Y\in\Gamma :  y_{i}\textstyle\prod_{y\in Y}y=z\text{ for some }y_{i}\in Y\right \}\right | \\
				&= |\Omega|/{(q-1)} \\
				&= {(q-2)(q-4)(q-5)}/{4!}.
			\end{align*}
			This completes the proof.	
		\end{proof}

	\section{Summary and concluding remarks}
		This paper systematically bridges coding theory and combinatorial designs by constructing $q$-ary $t$-designs from linear codes.
		To some extent, this is a deeper investigation of constructing ordinary designs from linear codes.
		We have consolidated known methods, provided concrete classifications, and introduced novel proofs for infinite families of codes. Our principal contributions are as follows.
		
		First, we refined classical results into two unified, general criteria: a Standard Criterion based on four fundamental parameters of a code, and a Puncturing–Shortening Criterion for deriving designs through code operations. The relationship between $q$-ary $t$-designs and $t$-regular codes have also been proposed.
		
		Second, we established new characterizations linking code properties directly to design structures: a linear code is MDS if and only if its minimum weight codewords form a specific $q$-ary $1$-design; a linear code with minimum distance $2e+1$ is perfect if and only if its minimum weight codewords form a $q$-ary $(e+1)$-design with $\lambda=1$.
		We presented a complete classification for codes with few weights: under monomial equivalence, simplex codes are the only one-weight codes yielding $q$-ary $2$-designs; we listed all five known classes of two-weight codes with dual distance $d^\perp \ge 4$ that produce $q$-ary $t$-designs; and we confirmed that extremal Type III and Type IV self-dual codes are a source of $q$-ary $t$-designs.
		
		Third, for codes falling outside the scope of the two criteria, we developed a new method leveraging the automorphism groups. Using this approach, we proved the following:
		\begin{itemize}
			\item[1.] A doubly-extended Reed-Solomon code $[q+1,k,q-k+2]_q$ yields a $q$-ary $2$-$(q+1,q-k+1,\lambda)$ design if $\gcd(k-1, q-1)=1$, where $\lambda = \binom{q-1}{k-1}/(q-1)$.
			\item[2.] For $q=2^m$ with $m \ge 5$, the trace codes $\mathcal{C}_{\{1,2,3\}}$ with parameters $[q+1,6,q-5]_q$ yield two infinite families of $q$-ary $2$-designs:
			$q$-ary $2$-$(q+1,q-5,\lambda_{1})$ designs and $q$-ary $2$-$(q+1,q-4,\lambda_{2})$ designs, where $\lambda_{1} = (q-2)(q-5)(q-6)(q-8)/6!$, $\lambda_{2} = (q-2)(q-4)(q-5)/4! $.
		\end{itemize}
		
		For $t\ge2$, the following interesting open problems arise naturally:
		\begin{itemize}
			\item [1.] Do codewords of other weights in the two families above also form $q$-ary $t$-designs? Moreover, do their dual codes hold $q$-ary $t$-designs?
			\item [2.] Whether extended generalized Reed-Solomon  codes, a generalization of doubly-extended Reed–Solomon codes, can give rise to $q$-ary $t$-designs?
			\item [3.] Does there exist an infinite family of linear codes holding $q$-ary $t$-designs with $t\ge3$?
		\end{itemize}

	\end{document}